\documentclass[10pt,twocolumn,longbibliography]{article}
\usepackage{dcolumn}
\usepackage{bm}
\usepackage[utf8x]{inputenc}
\usepackage[english]{babel}
\usepackage[T1]{fontenc}
\usepackage{lmodern}
\usepackage{amsmath,amsfonts,amssymb,amsthm,bm,times,dcolumn}
\usepackage{mwe}
\usepackage{microtype}
\usepackage{braket}
\usepackage{mathrsfs}
\usepackage{nicematrix}
\usepackage{array}
\usepackage{tikz}
\usepackage{blkarray}
\usepackage[colorlinks={true}, citecolor={blue}, filecolor={blue}, linkcolor={blue}, urlcolor={blue}]{hyperref}
\usepackage{graphicx,epstopdf,color}
\usepackage{subfigure}
\usepackage[mathscr]{euscript}
\usepackage{setspace}
\usepackage{physics}
\usepackage{booktabs}
\newtheorem{theorem}{Theorem}

\newtheorem{definition}{Definition}
\newtheorem{remark}{Remark}
\newtheorem{lemma}{Lemma}
\newtheorem{proposition}{Proposition}

\begin{document}

\title{Automorphism-Induced Entanglement Bounds in Many-Body Systems}

\author{
Saikat Sur \\
\texttt{saikats@imsc.res.in} \\
Optics \& Quantum Information Group \\
The Institute of Mathematical Sciences, HBNI \\
CIT Campus, Taramani, Chennai 600113, India
}

\twocolumn[
\maketitle
\begin{abstract}
We derive an upper bound on the maximum balanced bipartite entanglement entropy of ground states of many-body Hamiltonians defined on a graph, agnostic to any particular model, that possesses a nontrivial automorphism group. We show that the entropy is bounded by the logarithm of a weighted sum of multiplicities of irreducible representations of the bipartition-preserving automorphism subgroup. This bound complements the known degeneracy-based bound, with neither universally dominating the other. For the complete graph 
$K_n$, the new bound yields an exponential improvement from linear to logarithmic scaling in the system size, consistent with the exact value of the entropy.
\end{abstract}
\vspace{0.5cm}
]

\maketitle

\section{Introduction}
\label{sec:intro}

\label{sec:intro}

Many-body ground states are natural
reservoirs of quantum entanglement.
In low-energy physics, ground-state
entanglement underlies phenomena such
as topological order, superconductivity,
and spin-liquid
phases~\cite{sachdev, lindgren},
making these states carriers of
non-classical correlations that
distinguish quantum phases of matter.
At the same time, such correlations
can be harnessed for quantum information
processing tasks including quantum
computing~\cite{loss_1998,
benjamin2003quantum},
simulation~\cite{georgescu_2014,
johnson2011quantum}, and
communication~\cite{wang, kimble2008,
bose}, provided the entanglement is
sufficiently robust and operationally
accessible~\cite{kurizki_2015_pnas}.

A central theme in many-body physics
is the relationship between the interactions and the
entanglement properties of low-energy
eigenstates. For systems defined on
lattices, a foundational
result is the area law. It states that the entanglement
entropy across a bipartition scales
with the boundary area rather than the
volume for gapped local
Hamiltonians~\cite{Eisert2010,
hastings2007, vidal_2003}.
Beyond the area law, the
entanglement structure depends on the
geometry and symmetry of the underlying
interaction. Translational and point symmetries have been shown to
constrain correlation functions and
entanglement spectra of
eigenstates~\cite{subrah-heisen,
gnatenko_pla_2021, jafari_cjp_2022}.
Graph-theoretic frameworks have found
 applications in describing quantum
correlations, particularly in the context of multipartite
entanglement in graph states ~\cite{hein_2004,
hein_2006}, measurement-based
quantum computation~\cite{raussendorf_prl_2001},
and  entanglement structure and dynamics in
many-body
dynamics~\cite{toth_pra_2012,
prakash2024, saikat_npj}.
However, the role of \emph{graph
automorphism symmetry}, the
combinatorial symmetry group of the
interaction graph, in constraining the
entanglement of many-body ground states
remains comparatively unexplored.

The entanglement entropy of a pure
state across a balanced bipartition is
bounded above by the Schmidt rank of
the ground state, which in turn is
bounded by the sum of Schmidt ranks
of the ground space basis vectors.
This yields a general \emph{degeneracy-based
bound} (Proposition~\ref{prop:general_bound}) This bound for a ground state
requires microscopic knowledge of the
ground state  manifold for Hamiltonians. But 
the central question we address is:
\emph{can the automorphism group of the
interaction graph be used to derive a
geometry-based upper bound on the
ground state entanglement entropy,
requiring no microscopic knowledge of
the ground states?} The answer is
affirmative. Our bound 
(Theorem~\ref{thm:main}) depends
only on the geometry of the graph and
the representation theory of its
automorphism group, and applies to any
Hamiltonian commuting with the
automorphism group, including Ising,
Heisenberg, $XY$, and Kitaev models.

 To make the framework concrete, we use
the antiferromagnetic Ising Hamiltonian and the XY model
as primary
examples throughout. The  ground space of the Ising model
is spanned by the max-cut configurations
of the underlying graph~\cite{Barahona1982,
Papadimitriou1991}. For the XY model embedded on a graph, the exact entropy can be computed from numerical methods. We show that the automorphism-based bound and the degeneracy-based bounds are complementary.
The degeneracy-based bound is tighter when the ground state degeneracy
is small, while the automorphism-based
bound is tighter when the automorphism
group is large. 
Beyond the main bound, we establish
that the entanglement entropy of the  degenerate ground state  symmetrized with respect to the graph symmetries is bounded by the Hilbert space dimension.

For asymmetric graphs, which constitute almost all large
graphs~\cite{erdos_1963, Cameron1999}
, Theorem~\ref{thm:main} is
inapplicable and
Proposition~\ref{prop:general_bound}
is the only available bound. We analyze
this regime separately in
Section~\ref{sec:asymmetric}, showing
that for classical Hamiltonians the
bound is controlled by the residual
entropy per spin, while for generic
quantum Hamiltonians the problem of
obtaining a non-trivial bound is
computationally intractable, at least
as hard as
QCMA~\cite{gharibian2024,
bouland2024, weggemans_2025}.

The article is organized as follows.
Section~\ref{sec:why_balanced_cut_vne}
introduces the balanced-cut von Neumann
entropy, derives the degeneracy-based
bound (Proposition~\ref{prop:general_bound}).
In Section~\ref{sec:graph_theory}, we introduce the graph automorphism
framework, establish the symmetry
structure of the ground state
(Lemmas~\ref{lem:psi_invariant}--\ref{lem:intertwine}),
apply Schur's lemma to derive
Theorem~\ref{thm:main}, and prove other results related to graph symmetries.
Section~\ref{sec:examples} works out
the examples of $C_n$, $K_n$ and $K_{3,3}$.
Section~\ref{sec:asymmetric} treats
asymmetric graphs.
Section~\ref{sec:conclusion} discusses
open problems and future directions.

\section{Balanced-cut von Neumann entropy as the primary measure}
\label{sec:why_balanced_cut_vne}

Consider a quantum network of $N$ qubits
arranged on a graph $\Lambda = (V, E)$,
where each vertex $v \in V$ carries a
qubit and each edge $(i,j) \in E$
represents an interaction. A 
bipartition $V = A \sqcup B$ divides the network
into two disjoint sets of qubits $A$
and $B$. The von Neumann
entropy $S_{N/2}(|\psi\rangle)$ across
this cut quantifies the bipartite
entanglement between the two disjoint regions of
the network for any pure state
$|\psi\rangle$. The von Neumann entropy across a
bipartition $A|B$ is defined as 
\begin{equation}
\label{eq:balanced_ee}
    S_{N/2}(|\psi\rangle)
    = -\mathrm{tr}\bigl(
    \rho_A\log\rho_A\bigr),
    \qquad
    \rho_A = \mathrm{tr}_B\bigl(
    |\psi\rangle\langle\psi|\bigr),
\end{equation}

For the balanced bipartition $|A| = |B| = N/2$, the measure
provides the largest possible bipartite entanglement capacity across cuts. The entanglement
entropy is bounded above by
$\frac{N}{2}\log 2$, a quantity that
grows extensively with system size $N$,
making it a thermodynamic
measure of entanglement capacity~\cite{nielsen_chaung,page_1993}. In contrast,  for
unequal bipartitions $|A| < |B|$, the
entropy is bounded by $|A|\log 2$,
which is limited by the size of the smaller subsystem.
The balanced cut therefore probes the
maximum bipartite entanglement capacity
of the network~\cite{horodecki_rmp_2009,vardoyan_2023},
in the sense that it maximizes the
potential Schmidt rank $\min(\dim
\mathcal{H}_A, \dim\mathcal{H}_B)
= 2^{N/2}$ over all bipartitions. Therefore it is the natural diagnostic of bipartite
entanglement structure in many-body
physics~\cite{Amico2008, Eisert2010}.
Its scaling distinguishes the dominant
phases: for gapped local Hamiltonians
the entropy satisfies an {area law}
$S = \mathcal{O}(|\partial A|)$~\cite{Eisert2010,
hastings2007}; at quantum critical
points it grows as $S = \mathcal{O}(\log
N)$~\cite{vidal_2003, Calabrese2004, rigol_nature_2008} in one dimension;
and in volume-law phases $S = \mathcal{O}(N)$.
These scalings underpin the efficiency
of tensor network methods such as
DMRG~\cite{White1992}, the diagnosis
of many-body
localization~\cite{Bauer2013, nandkishore_2015}, and
thermalization~\cite{DAlessio2016}.

\subsection{Setup}

Fix a balanced bipartition $V = A \sqcup B$ with
$|A| = |B| = N/2$ ($N$ even throughout; the odd case is
analogous with $|A| = \lfloor N/2 \rfloor$).
The Hilbert space factorizes as
$\mathcal{H} = \mathcal{H}_A \otimes \mathcal{H}_B$
with
$\mathcal{H}_A = \mathcal{H}_B
= (\mathbb{C}^2)^{\otimes N/2}$.
Any state $|\psi\rangle \in \mathcal{H}$ expands in the
computational basis as
\begin{equation}
\label{eq:psi_expand}
    |\psi\rangle
    = \sum_{\alpha}\sum_{\beta}
    M_{\alpha\beta}\,|\alpha\rangle_A|\beta\rangle_B,
\end{equation}
where $\alpha,\beta \in \{+1,-1\}^{N/2}$ label spin
configurations on $A$ and $B$ respectively, and
$M_{\alpha\beta} = \langle\alpha,\beta|\psi\rangle$
is the \emph{coefficient matrix}, a
$2^{N/2}\times 2^{N/2}$ complex matrix.
The singular value decomposition of $M$ yields the Schmidt
decomposition of $|\psi\rangle$: the Schmidt coefficients
$\{\lambda_k\}$ are the  singular values of $M$,
satisfying $\lambda_k \geq 0$, $\sum_k\lambda_k = 1$, and
\begin{equation}
\label{eq:entropy_schmidt}
    S_{N/2}(|\psi\rangle)
    = -\sum_k \lambda_k\log\lambda_k
    \leq \log r,
\end{equation}
where $r = \mathrm{rank}(M)$ is the Schmidt rank of
$|\psi\rangle$ across $A|B$~\cite{nielsen_chaung}.

Let $\mathcal{G}$ denote the ground space of the
Hamiltonian with orthonormal basis
$\{|g_k\rangle\}_{k=1}^{d}$, where $d$ is the ground
state degeneracy. Any state in $\mathcal{G}$ takes the
form
$|\psi\rangle = \sum_{k=1}^{d} c_k |g_k\rangle$
with $c_k \in \mathbb{C}$ and $\sum_k |c_k|^2 = 1$.
Define the coefficient matrices
\begin{equation}
\label{eq:Mk}
    M^{(k)}_{\alpha\beta}
    = \langle\alpha,\beta|g_k\rangle,
    \qquad k = 1,\ldots,d,
\end{equation}
so that the coefficient matrix of $|\psi\rangle$
decomposes as
\begin{equation}
\label{eq:M_decomp}
    M = \sum_{k=1}^{d} c_k M^{(k)}.
\end{equation}

\begin{lemma}
\label{lem:rank_subadd}
Let $|\psi\rangle \in \mathcal{G}$ with coefficient
matrix $M = \sum_{k=1}^d c_k M^{(k)}$ as
in~\eqref{eq:M_decomp}. Then the bound from rank subadditivity is given as
\begin{equation}
\label{eq:rank_bound}
    \mathrm{rank}(M)
    \leq \sum_{k=1}^{d} \mathrm{rank}(M^{(k)}).
\end{equation}
\end{lemma}

\begin{proof}
Subadditivity of matrix rank gives
$\mathrm{rank}(\sum_k c_k M^{(k)})
\leq \sum_k \mathrm{rank}(c_k M^{(k)})
= \sum_k \mathrm{rank}(M^{(k)})$,
where the last equality holds since $c_k \neq 0$
scaling does not change rank.
\end{proof}

\begin{proposition}
\label{prop:general_bound}
Let $r_k = \mathrm{rank}(M^{(k)})$ for $k=1,\ldots,d$.
For any state $|\psi\rangle$ in the ground space
$\mathcal{G}$, the general entropy bound is
\begin{equation}
\label{eq:general_bound}
    S_{N/2}(|\psi\rangle)
    \leq \log \sum_{k=1}^{d} r_k.
\end{equation}
\end{proposition}

\begin{proof}
From~\eqref{eq:entropy_schmidt} and
Lemma~\ref{lem:rank_subadd},
\begin{equation}
    S_{N/2}(|\psi\rangle)
    \leq \log r
    \leq \log \sum_{k=1}^{d} r_k,
\end{equation}
where the second inequality follows from
$r \leq \sum_k r_k$ and monotonicity of $\log$.
\end{proof}

\begin{remark}
\label{rem:nontrivial}
The bound~\eqref{eq:general_bound} is non-trivial
precisely when $\sum_k r_k < 2^{N/2}$, i.e.\ when
the rank subadditivity constraint is tighter
than the Hilbert space dimension.
The quantity $\sum_k r_k$ requires independent
microscopic knowledge of the ground states and cannot
be inferred from the entropy alone. For graphs with
trivial automorphism group $\mathrm{Aut}(\Lambda) = \{e\}$,
which we call \emph{asymmetric graphs},
Proposition~\ref{prop:general_bound} generically
reduces to the trivial bound; we treat this class
separately in Section~\ref{sec:asymmetric}. A
non-trivial improvement can be obtained by exploiting
automorphism symmetry, as established in the next 
Section.
\end{remark}

\section{Graph Automorphisms and Ground State Entanglement}
\label{sec:graph_theory}

Let $\Lambda = (V, E)$ be a finite graph on $N$ vertices with edge set $E$.
The \emph{automorphism group} of $\Lambda$ is~\cite{godsil_royle, grohe_10}
\begin{equation}
\label{eq:G_def}
    G = \mathrm{Aut}(\Lambda)
    = \bigl\{g : V \to V \;\big|\;
    (i,j) \in E \Leftrightarrow (g(i),g(j)) \in E\bigr\},
\end{equation}
the group of all vertex permutations preserving the edge structure.
Each $g \in G$ acts on spin configurations by permuting sites,
inducing a unitary operator on the Hilbert space
\begin{equation}
\label{eq:Ug_def}
    U(g)|\mathbf{s}\rangle = |g \cdot \mathbf{s}\rangle,
    \qquad \mathbf{s} \in \{+1,-1\}^N.
\end{equation}
Since $g$ preserves the edge set, it commutes with any Hamiltonian
$H(\Lambda)$ built from interactions along edges of $\Lambda$,
\begin{equation}
\label{eq:commute}
    [H(\Lambda),\, U(g)] = 0,
    \qquad \forall\, g \in G.
\end{equation}

Given a balanced bipartition $V = A \sqcup B$ with $|A| = |B| = N/2$,
we restrict to those automorphisms that preserve the bipartition,
\begin{equation}
\label{eq:GammaA_def}
    \Gamma_A = \{g \in G : g(A) = A\}.
\end{equation}
This is a subgroup of $G$: the identity preserves $A$; if
$g, h \in \Gamma_A$ then $g \circ h$ and $g^{-1}$ also preserve $A$. Since $V = A \sqcup B$ and $g$ is a bijection on $V$,
the condition $g(A) = A$ automatically implies $g(B) = B$.
The reason for restricting to $\Gamma_A$ rather than the full $G$ is
the factorization property: for $g \in \Gamma_A$, since $g$ maps
$A$-sites to $A$-sites and $B$-sites to $B$-sites independently,
the unitary $U(g)$ factorizes as
\begin{equation}
\label{eq:U_factorize}
    U(g) = P_A(g) \otimes P_B(g),
    \qquad g \in \Gamma_A,
\end{equation}
where $P_A(g)$ and $P_B(g)$ are the unitary permutation
representations induced by $g$ on $\mathcal{H}_A$ and
$\mathcal{H}_B$ respectively,
\begin{equation}
\label{eq:PA_PB_def}
    P_A(g)|\alpha\rangle_A = |g(\alpha)\rangle_A,
    \qquad
    P_B(g)|\beta\rangle_B = |g(\beta)\rangle_B.
\end{equation}
Here $g(\alpha)$ denotes the spin configuration on $A$ obtained
by the site permutation $g$: if
$\alpha = (s_{i_1}, \ldots, s_{i_{N/2}})$ with
$A = \{i_1, \ldots, i_{N/2}\}$, then
$[g(\alpha)]_k = s_{g^{-1}(i_k)}$.

\subsection{Orbits and symmetry-inequivalent
configurations}
The full automorphism group $G$ acts on the
complete spin configuration space $\{0,1\}^N$,
partitioning it into orbits
\begin{equation}
    \{0,1\}^N = \mathcal{C}_1 \sqcup \mathcal{C}_2
    \sqcup \cdots \sqcup \mathcal{C}_\omega,
\end{equation}
where $\omega$ denotes the number of $G$-orbits
on $\{0,1\}^N$. Restricting to the bipartition,
$\Gamma_A \leq G$ acts on the set of
$A$-configurations $\{0,1\}^{N/2}$, partitioning
it into orbits
\begin{equation}
\label{eq:A_orbits}
    \{0,1\}^{N/2}
    = \mathcal{O}_1^A \sqcup \mathcal{O}_2^A
    \sqcup \cdots \sqcup \mathcal{O}_{\omega_A}^A,
\end{equation}
where $\omega_A$ denotes the number of
$\Gamma_A$-orbits on $\{0,1\}^{N/2}$.
The number of orbits $\omega_A$ is given by
Burnside's lemma~\cite{Burnside1897,Cameron1999},
\begin{equation}
\label{eq:burnside}
    \omega_A = \frac{1}{|\Gamma_A|}
    \sum_{g \in \Gamma_A}
    \bigl|\{\alpha \in \{0,1\}^{N/2} :
    g(\alpha) = \alpha\}\bigr|,
\end{equation}
i.e., the number of orbits equals the average
number of $A$-configurations fixed by an element
of $\Gamma_A$.

\subsection{Symmetrized ground states}
The orbit-averaged states are given by the following form
\begin{equation}
\label{eq:A1_basis}
    |\phi_k\rangle
    = \frac{1}{\sqrt{|\mathcal{C}_k|}}
    \sum_{\mathbf{s} \in \mathcal{C}_k}
    |\mathbf{s}\rangle,
    \qquad k = 1, \ldots, \omega,
\end{equation}
transform in the trivial irreducible
representation $A_1$ of $G$. Since the orbits
$\mathcal{C}_1,\ldots,\mathcal{C}_\omega$ are
disjoint, the states $|\phi_k\rangle$ have
disjoint support in the computational basis and
are therefore orthonormal. They span the
trivial-sector subspace of $\mathcal{H}$ as
\begin{equation}
\label{eq:HA1_def}
    \mathcal{H}_{A_1}(\Lambda)
    = \mathrm{span}\bigl\{
    |\phi_k\rangle :
    k = 1, \ldots, \omega\bigr\},
    \qquad
    \dim\mathcal{H}_{A_1}(\Lambda) = \omega.
\end{equation}

\begin{lemma}
\label{lem:psi_invariant}
Every state $|\psi\rangle \in
\mathcal{H}_{A_1}(\Lambda)$ satisfies
\begin{equation}
    U(g)|\psi\rangle = |\psi\rangle
    \qquad \forall\, g \in G.
\end{equation}
\end{lemma}
\begin{proof}
Any $|\psi\rangle \in \mathcal{H}_{A_1}(\Lambda)$
takes the form
$|\psi\rangle = \sum_{k=1}^\omega
c_k|\phi_k\rangle$.
For any $g \in G$, since $\mathcal{C}_k$ is a
$G$-orbit, the map
$\mathbf{s} \mapsto g\cdot\mathbf{s}$ is a
bijection from $\mathcal{C}_k$ to itself.
Therefore
\begin{equation}
    U(g)|\phi_k\rangle
    = \frac{1}{\sqrt{|\mathcal{C}_k|}}
    \sum_{\mathbf{s} \in \mathcal{C}_k}
    |g\cdot\mathbf{s}\rangle
    = |\phi_k\rangle.
\end{equation}
Since each basis state is individually
invariant, so is any
$|\psi\rangle = \sum_k c_k|\phi_k\rangle$.
\end{proof}

{ 
\noindent\textit{Physical motivation for the
symmetrized ground state.}
Since $[H(\Lambda), U(g)] = 0$ for all $g \in G$, the Gibbs state
$\rho_\beta \propto e^{-\beta H(\Lambda)}$ satisfies
\begin{equation}
    U(g)\rho_\beta U(g)^\dagger = \rho_\beta
    \qquad \forall\,g \in G,\beta \geq 0
\end{equation}
as an exact algebraic identity, since
$U(g)e^{-\beta H(\Lambda)}U(g)^\dagger
= e^{-\beta H(\Lambda)}$ follows immediately from
$[H(\Lambda),U(g)]=0$. As $\beta\to\infty$,
$\rho_\beta$ converges to the normalized projector
onto the ground space,
\begin{equation}
    \rho_0 = \frac{1}{d}\,\Pi_{\mathcal{G}(\Lambda)},
\end{equation}
which inherits $G$-invariance from the Gibbs state
at all finite temperatures.

For the non-degenerate case ($d=1$), 
the ground state $|\psi_0\rangle$ is already
a pure $G$-invariant pure state. It satisfies
$U(g)|\psi_0\rangle = e^{i\phi(g)}|\psi_0\rangle$
for some phase $e^{i\phi(g)}$. Since this phase
factors out of the intertwining condition of
Lemma~\ref{lem:intertwine} and does not affect
the block structure of $M$,
Theorem~\ref{thm:main} applies directly without
any further assumption.

For $d > 1$,
$\rho_0$ is a mixed state representing a statistical
ensemble over the ground space rather than a single
pure state. The $G$-invariance of the mixed state $\rho_0$ does
not imply that the pure states in any decomposition
of $\rho_0$ are individually $G$-invariant, 
a uniform mixture over an orbit of $G$ is
$G$. One might therefore attempt to work directly with
the mixed state $\rho_0$ and derive an entanglement
bound from its $G$-invariance. However, as we show
in Theorem~\ref{thm:mixed_bound}, applying
Schur's lemma to the reduced density matrix
$R = \mathrm{tr}_B(\rho_0)$ yields only the trivial
bound
$    S_{N/2}(\rho_0)
    \leq \tfrac{N}{2}\log 2$, 
which coincides with the Hilbert space bound.
 }

\begin{lemma}
\label{lem:M_symmetry}
For any $|\psi\rangle \in \mathcal{G}_{A_1}(\Lambda)$ and
any $g \in \Gamma_A$, the coefficient matrix $M$ satisfies
\begin{equation}
\label{eq:M_sym}
    M_{g(\alpha),\, g(\beta)} = M_{\alpha\beta}
    \qquad \forall\, \alpha, \beta \in \{+1,-1\}^{N/2}.
\end{equation}
That is, $M$ is constant on the joint $(\alpha,\beta)$-orbits
of $\Gamma_A$.
\end{lemma}

\begin{proof}
From Lemma~\ref{lem:psi_invariant} and the
factorization~\eqref{eq:U_factorize},
\begin{eqnarray}
    &&\sum_{\alpha,\beta} M_{\alpha\beta}
    |g(\alpha)\rangle_A |g(\beta)\rangle_B
    = U(g)|\psi\rangle = |\psi\rangle \nonumber\\
    &=& \sum_{\alpha,\beta} M_{\alpha\beta}
    |\alpha\rangle_A|\beta\rangle_B.
\end{eqnarray}
Substituting $\alpha' = g(\alpha)$, $\beta' = g(\beta)$
on the left and comparing coefficients in the orthonormal
basis $\{|\alpha'\rangle_A|\beta'\rangle_B\}$ gives
$M_{g^{-1}(\alpha'), g^{-1}(\beta')} = M_{\alpha'\beta'}$.
Replacing $\alpha' \to g(\alpha)$, $\beta' \to g(\beta)$
yields~\eqref{eq:M_sym}.
\end{proof}

\begin{lemma}
\label{lem:intertwine}
For any $|\psi\rangle \in \mathcal{G}_{A_1}(\Lambda)$ and
any $g \in \Gamma_A$, the coefficient matrix $M$ satisfies
\begin{equation}
\label{eq:intertwine}
    P_A(g)\, M = M\, P_B(g).
\end{equation}
\end{lemma}

\begin{proof}
We compute both sides entry by entry.

\medskip
\noindent\textit{Left side.}
By matrix multiplication,
\begin{equation}
    [P_A(g)M]_{\alpha\beta}
    = \sum_{\alpha''} [P_A(g)]_{\alpha\alpha''}\,
    M_{\alpha''\beta}.
\end{equation}
Since $P_A(g)$ is the permutation matrix induced by $g$
on $\mathcal{H}_A$, its entries are
\begin{equation}
    [P_A(g)]_{\alpha\alpha''}
    = \langle\alpha|P_A(g)|\alpha''\rangle
    = \langle\alpha|g(\alpha'')\rangle
    = \delta_{\alpha,\, g(\alpha'')},
\end{equation}
which equals unity if and only if $\alpha'' = g^{-1}(\alpha)$
and vanishes otherwise. Therefore only one term survives
the sum,
\begin{equation}
\label{eq:lhs}
    [P_A(g)M]_{\alpha\beta}
    = M_{g^{-1}(\alpha),\,\beta}.
\end{equation}

\medskip
\noindent\textit{Right side.}
By matrix multiplication,
\begin{equation}
    [MP_B(g)]_{\alpha\beta}
    = \sum_{\beta''} M_{\alpha\beta''}\,
    [P_B(g)]_{\beta''\beta}.
\end{equation}
Similarly, the entries of $P_B(g)$ are
\begin{equation}
    [P_B(g)]_{\beta''\beta}
    = \langle\beta''|P_B(g)|\beta\rangle
    = \langle\beta''|g(\beta)\rangle
    = \delta_{\beta'',\, g(\beta)},
\end{equation}
which equals unity if and only if $\beta'' = g(\beta)$
and vanishes otherwise. Therefore
\begin{equation}
\label{eq:rhs}
    [MP_B(g)]_{\alpha\beta}
    = M_{\alpha,\, g(\beta)}.
\end{equation}

\medskip
\noindent\textit{Equality.}
From Lemma~\ref{lem:M_symmetry},
$M_{g(\alpha), g(\beta)} = M_{\alpha\beta}$
for all $\alpha, \beta$ and all $g \in \Gamma_A$.
Substituting $\alpha \mapsto g^{-1}(\alpha)$,
\begin{equation}
    M_{g(g^{-1}(\alpha)),\, g(\beta)}
    = M_{g^{-1}(\alpha),\, \beta}
    \quad\Rightarrow\quad
    M_{\alpha,\, g(\beta)}
    = M_{g^{-1}(\alpha),\, \beta}.
\end{equation}
Comparing with~\eqref{eq:lhs} and~\eqref{eq:rhs},
\begin{equation}
    [P_A(g)M]_{\alpha\beta}
    = M_{g^{-1}(\alpha),\,\beta}
    = M_{\alpha,\, g(\beta)}
    = [MP_B(g)]_{\alpha\beta}
\end{equation}
for all $\alpha, \beta$. Hence $P_A(g)M = MP_B(g)$.
\end{proof}

\subsection{Irreducible decomposition}

Equation~\eqref{eq:intertwine} states that $M$ intertwines
the permutation representation of $\Gamma_A$ on
$\mathcal{H}_A$ with that on $\mathcal{H}_B$. We exploit
this to bound $\mathrm{rank}(M)$ via the irreducible
decomposition of these representations. The collection $\{P_A(g) : g \in \Gamma_A\}$ forms a
unitary permutation representation of $\Gamma_A$ on
$\mathcal{H}_A$. It decomposes into irreducible
representations as
\begin{equation}
\label{eq:HA_decomp}
    \mathcal{H}_A
    = \bigoplus_\mu
    \left(\mathbb{C}^{m_\mu^A} \otimes \mathcal{V}_\mu\right),
    \qquad
    \sum_\mu d_\mu m_\mu^A = 2^{N/2},
\end{equation}
where $\mu$ runs over the distinct irreps of $\Gamma_A$,
$\mathcal{V}_\mu$ is the carrier space of irrep $\mu$ with
dimension $d_\mu$, and $m_\mu^A$ is the multiplicity with
which irrep $\mu$ appears in $\mathcal{H}_A$.
Similarly, $\{P_B(g) : g \in \Gamma_A\}$ is a unitary
permutation representation of $\Gamma_A$ on $\mathcal{H}_B$,
decomposing as
\begin{equation}
\label{eq:HB_decomp}
    \mathcal{H}_B
    = \bigoplus_\mu
    \left(\mathbb{C}^{m_\mu^B} \otimes \mathcal{V}_\mu\right),
    \qquad
    \sum_\mu d_\mu m_\mu^B = 2^{N/2},
\end{equation}
where $m_\mu^B$ is the multiplicity of irrep $\mu$ in
$\mathcal{H}_B$. The index $\mu$ runs over the same set
of irreps of $\Gamma_A$ in both decompositions; but some
multiplicities may be zero in either $\mathcal{H}_A$
or $\mathcal{H}_B$.

\begin{definition}
\label{def:sym_basis}
Let $\Gamma_A$ act on $\mathcal{H}_A$ via the unitary
permutation representation $\{P_A(g) : g \in \Gamma_A\}$.
Let $\hat{\Gamma}_A$ denote the set of equivalence classes
of irreducible representations of $\Gamma_A$. For each
$\mu \in \hat{\Gamma}_A$, let $\mathcal{V}_\mu$ be the
carrier space of irrep $\mu$ with dimension $d_\mu$, and
let $P_\mu(g)$ denote the restriction of $P_A(g)$ to
$\mathcal{V}_\mu$. The character of irrep $\mu$ at
$g \in \Gamma_A$ is
\begin{equation}
\label{eq:character}
    \chi_\mu(g) = \mathrm{tr}(P_\mu(g)).
\end{equation}
The orthogonal projector onto the $\mu$-sector of
$\mathcal{H}_A$ is~\cite{Serre1977, unger_2006}
\begin{equation}
\label{eq:proj}
    \Pi_\mu^A = \frac{d_\mu}{|\Gamma_A|}
    \sum_{g \in \Gamma_A} \chi_\mu(g)^*\, P_A(g),
\end{equation}
satisfying $\Pi_\mu^A|\mu,i,e_a\rangle
= |\mu,i,e_a\rangle$ and
$\Pi_\mu^A|\nu,j,e_b\rangle = 0$ for $\nu \neq \mu$.
The \emph{symmetry-adapted basis} of $\mathcal{H}_A$
consists of states
\begin{equation}
\label{eq:sym_basis}
    |\mu, i, e_a\rangle,
    \qquad
    \mu \in \hat{\Gamma}_A,
    \quad
    i = 1,\ldots,m_\mu^A,
    \quad
    a = 1,\ldots,d_\mu,
\end{equation}
where $i$ labels the $m_\mu^A$ independent copies of
sector $\mu$ in $\mathcal{H}_A$, and $e_a$ labels a
basis of $\mathcal{V}_\mu$. These states form a complete
orthonormal basis of $\mathcal{H}_A$,
\begin{equation}
\label{eq:completeness}
    \sum_{\mu \in \hat{\Gamma}_A}
    \sum_{i=1}^{m_\mu^A}
    \sum_{a=1}^{d_\mu}
    |\mu, i, e_a\rangle
    \langle\mu, i, e_a|
    = I_{\mathcal{H}_A},
\end{equation}
and analogously for $\mathcal{H}_B$ with multiplicities
$m_\mu^B$ and basis states $|\mu, j, e_b\rangle$,
$j = 1,\ldots,m_\mu^B$, $b = 1,\ldots,d_\mu$.
\end{definition}

\begin{remark}
\label{rem:three_labels}
The three labels $(\mu, i, e_a)$ correspond
to three physically distinct and mutually commuting
sets of observables. The sector label $\mu$ is
determined by the graph geometry through $\Gamma_A$.
The multiplicity label $i$ is determined by
additional physical conserved quantities of
$H(\Lambda)$ that are independent of the geometry.
The internal label $e_a$ is determined by the
action of $\Gamma_A$ within the irrep $\mu$.
Since $\Pi_\mu^A$ is a linear combination of
$P_A(g)$, it commutes with all $P_A(g)$ and hence
with $O_i$. Within sector $\mu$, $O_i$
acts on the multiplicity space $\mathbb{C}^{m_\mu^A}$
and $O_{e_a}$ acts on the internal space
$\mathcal{V}_\mu$ and operators on different tensor
factors commute trivially. The three observables
therefore commute pairwise and can be simultaneously
diagonalized.
\end{remark}

\begin{lemma}
\label{lem:schur}
Let $M$ be the coefficient matrix satisfying the
intertwining condition~\eqref{eq:intertwine}, and
let $M_{\mu\nu} = \Pi_\mu^A\, M\, \Pi_\nu^B$ denote
the block of $M$ mapping the $\nu$-sector of
$\mathcal{H}_B$ to the $\mu$-sector of $\mathcal{H}_A$.
Then:

\noindent(a) For $\mu \neq \nu$, the off-diagonal blocks vanish
\begin{equation}
\label{eq:schur_offdiag}
    M_{\mu\nu} = 0.
\end{equation}

\noindent(b)  Within the $\mu$-block, the coefficient matrix $M$ is block diagonal
\begin{equation}
\label{eq:schur_diag}
    M_{\mu\mu} = \Phi_\mu \otimes I_{d_\mu},
\end{equation}
where $\Phi_\mu$ is an $m_\mu^A \times m_\mu^B$
matrix acting on the multiplicity indices, and
$I_{d_\mu}$ is the identity on $\mathcal{V}_\mu$.
In terms of basis matrix elements,
\begin{equation}
    \langle\mu,{i},e_a|\,
    M_{\mu\mu}\,
    |\mu,{j},e_b\rangle
    = (\Phi_\mu)_{\mathbf{i}\mathbf{j}}\,
    \delta_{ab}.
\end{equation}

 \noindent(c) The rank of the coefficient Matrix is
\begin{equation}
\label{eq:schur_rank}
    \mathrm{rank}(M)
    = \sum_\mu d_\mu\,\mathrm{rank}(\Phi_\mu)
    \leq \sum_\mu d_\mu\,
    \min(m_\mu^A,\, m_\mu^B).
\end{equation}
\end{lemma}

\begin{proof}
\textit{(a)}
The intertwining condition~\eqref{eq:intertwine}
gives $P_A(g)M = MP_B(g)$ for all $g \in \Gamma_A$.
Projecting on the left by $\Pi_\mu^A$ and on the
right by $\Pi_\nu^B$,
\begin{equation}
    \Pi_\mu^A P_A(g)\, M_{\mu\nu}
    = M_{\mu\nu}\, P_B(g)\Pi_\nu^B
    \qquad \forall\, g \in \Gamma_A.
\end{equation}
Since $\Pi_\mu^A$ commutes with $P_A(g)$ and
projects onto the $\mu$-sector, $M_{\mu\nu}$
intertwines irrep $\nu$ of $\Gamma_A$ on
$\mathcal{H}_B$ with irrep $\mu$ of $\Gamma_A$
on $\mathcal{H}_A$. By Schur's
lemma~\cite{Serre1977}, any intertwiner between
two inequivalent irreducible representations is
zero. Hence $M_{\mu\nu} = 0$ for $\mu \neq \nu$.

\medskip
\noindent\textit{ (b)}
For $\mu = \nu$, $M_{\mu\mu}$ maps
$\mathbb{C}^{m_\mu^B} \otimes \mathcal{V}_\mu$
to $\mathbb{C}^{m_\mu^A} \otimes \mathcal{V}_\mu$,
satisfying
\begin{equation}
    (I_{m_\mu^A} \otimes P_\mu(g))\,
    M_{\mu\mu}
    = M_{\mu\mu}\,
    (I_{m_\mu^B} \otimes P_\mu(g))
    \qquad \forall\, g \in \Gamma_A.
\end{equation}
By Schur's lemma~\cite{Serre1977}, any intertwiner
of an irreducible representation with itself is a
scalar multiple of the identity on that representation
space. Therefore $M_{\mu\mu}$ acts as a scalar on
$\mathcal{V}_\mu$, with the scalar depending only on
the multiplicity indices $i,j$. That is,
\begin{equation}
    \langle\mu, i, e_a|\,
    M_{\mu\mu}\,
    |\mu, j, e_b\rangle
    = (\Phi_\mu)_{ij}\,\delta_{ab},
\end{equation}
for some $m_\mu^A \times m_\mu^B$ matrix $\Phi_\mu$.
In operator form, $M_{\mu\mu} = \Phi_\mu \otimes I_{d_\mu}$.

\medskip
\noindent\textit{ (c)}
Since the blocks $\{M_{\mu\mu}\}$ are mutually
orthogonal by part (a),
\begin{eqnarray}
    &&\mathrm{rank}(M)
    = \sum_\mu \mathrm{rank}(M_{\mu\mu}) 
    = \sum_\mu \mathrm{rank}
    (\Phi_\mu \otimes I_{d_\mu}) \nonumber\\
    &&= \sum_\mu d_\mu\,\mathrm{rank}(\Phi_\mu).
\end{eqnarray}
Since $\Phi_\mu$ is an $m_\mu^A \times m_\mu^B$
matrix, $\mathrm{rank}(\Phi_\mu)
\leq \min(m_\mu^A, m_\mu^B)$, giving
\begin{equation}
    \mathrm{rank}(M)
    \leq \sum_\mu d_\mu\,
    \min(m_\mu^A,\, m_\mu^B).
\end{equation}
\end{proof}

\begin{theorem}
\label{thm:main}
Let $\Lambda$ be a graph on $N$ vertices,
$G = \mathrm{Aut}(\Lambda)$, and $H(\Lambda)$
a Hamiltonian commuting with $U(g)$ for all
$g \in G$. Let $V = A \sqcup B$ be a balanced
bipartition with $|A| = |B| = N/2$, and let
$\Gamma_A = \{g \in G : g(A) = A\}$.
For any state $|\psi\rangle \in
\mathcal{G}_{A_1}(\Lambda)$,
\begin{equation}
\label{eq:main_bound}
    S_{N/2}(|\psi\rangle)
    \leq \log\!\left(
    \sum_\mu d_\mu\,
    \min(m_\mu^A,\, m_\mu^B)
    \right),
\end{equation}
where the sum is over irreps $\mu$ of $\Gamma_A$,
$d_\mu = \dim\mathcal{V}_\mu$, and $m_\mu^A$,
$m_\mu^B$ are the multiplicities of irrep $\mu$
in $\mathcal{H}_A$ and $\mathcal{H}_B$
respectively.
\end{theorem}

\begin{proof}
From Lemma~\ref{lem:psi_invariant} and
Lemma~\ref{lem:intertwine}, any
$|\psi\rangle \in \mathcal{G}_{A_1}(\Lambda)$
has coefficient matrix $M$ satisfying
$P_A(g)M = MP_B(g)$ for all $g \in \Gamma_A$.
By Lemma~\ref{lem:schur}(c),
\begin{equation}
    \mathrm{rank}(M)
    \leq \sum_\mu d_\mu\,
    \min(m_\mu^A,\, m_\mu^B).
\end{equation}
By Proposition~\ref{prop:general_bound}
and~\eqref{eq:entropy_schmidt},
$S_{N/2}(|\psi\rangle) \leq \log\,
\mathrm{rank}(M)$. Combining and applying
monotonicity of $\log$ gives~\eqref{eq:main_bound}.
\end{proof}

\begin{remark}
\label{rem:complementarity}
Proposition~\ref{prop:general_bound} and
Theorem~\ref{thm:main} are complementary.
Proposition~\ref{prop:general_bound} gives
the bound $S_{N/2} \leq \log\sum_k r_k$
in terms of the Schmidt ranks of individual
ground space basis vectors, requiring
microscopic knowledge of the ground states.
Theorem~\ref{thm:main} gives the bound
$S_{N/2} \leq \log\sum_\mu d_\mu
\min(m_\mu^A, m_\mu^B)$ entirely in terms
of the representation theory of $\Gamma_A$,
which is determined by the graph geometry
alone. When $\mathrm{Aut}(\Lambda) = \{e\}$,
the group $\Gamma_A$ is trivial, only the
trivial irrep appears, and
Theorem~\ref{thm:main} reduces to the
Hilbert space bound; we treat this case  separately
in Section~\ref{sec:asymmetric}.
\end{remark}

\begin{remark}
\label{rem:sandwich}
The block diagonal structure of
Lemma~\ref{lem:schur} yields not only
the upper bound of Theorem~\ref{thm:main}
but also a complementary lower bound.
The matrix $M$ is block diagonal with
blocks $M_{\mu\mu} = \Phi_\mu \otimes
I_{d_\mu}$. Let
$\hat{\mu} = \{\mu : \Phi_\mu \neq 0\}$
denote the set of irreps with non-zero
contribution to the ground state. The
entanglement entropy satisfies the
 bound
\begin{equation}
\label{eq:sandwich}
    \log \sum_{\mu \in \hat{\mu}} d_\mu
    \leq S_{N/2}(|\psi\rangle)
    \leq \log \sum_{\mu \in \hat{\mu}}
    d_\mu \min(m_\mu^A, m_\mu^B).
\end{equation}

The lower bound is achieved when each
block is minimally entangled, i.e.\
$\mathrm{rank}(\Phi_\mu) = 1$ for all
$\mu \in \hat{\mu}$. In this case each
block $\mu$ contributes exactly $d_\mu$
equal singular values to the Schmidt
decomposition. The upper bound is achieved
when each block is maximally entangled,
i.e.\ $\mathrm{rank}(\Phi_\mu) =
\min(m_\mu^A, m_\mu^B)$ for all
$\mu \in \hat{\mu}$, recovering
Theorem~\ref{thm:main}. The upper bound is tighter for
Hamiltonians whose ground state lies
in the symmetric sector
$\mathcal{G}_{A_1}(\Lambda)$, where
$\hat{\mu} = \{A_1\}$ and
$d_{A_1} = 1$ renders the lower bound
trivial zero while the upper
bound $\log \min(m_{A_1}^A, m_{A_1}^B)$
can be non-trivial. Conversely, the
lower bound is non-trivial when
$\Gamma_A$ is non-abelian and the
ground state has support in
high-dimensional irreps $\rho$ with
$d_\rho > 1$, a regime that arises for
non-diagonal Hamiltonians on highly
symmetric graphs. The examples in
Section~\ref{sec:examples} illustrate
the former regime; but finding
examples with non-trivial lower bounds
for physically relevant Hamiltonians
is left as an open problem.
\end{remark}

\begin{theorem}
\label{thm:mixed_bound}
Let $H(\Lambda)$ commute with $U(g)$ for all
$g \in G$, and let
$\rho_0 = \frac{1}{d}\Pi_{\mathcal{G}(\Lambda)}$
be the $G$-invariant mixed ground state.
Then the entanglement entropy of $\rho_0$
across the balanced bipartition $A|B$ satisfies
\begin{equation}
\label{eq:mixed_bound}
    S_{N/2}(\rho_0) \leq \frac{N}{2}\log 2.
\end{equation}
\end{theorem}

\begin{proof}

\noindent 
(a) Let $R = \mathrm{tr}_B(\rho_0)$ be the reduced
density matrix on $\mathcal{H}_A$. Since
$\rho_0$ is $G$-invariant, for any
$g \in \Gamma_A$ with
$U(g) = P_A(g)\otimes P_B(g)$
\begin{equation}
    [P_A(g)\otimes P_B(g)]\,\rho_0\,
    [P_A(g)^\dagger\otimes P_B(g)^\dagger]
    = \rho_0.
\end{equation}
Taking the partial trace over $B$ of both sides
and using the identity
$\mathrm{tr}_B([A\otimes U]\rho[A^\dagger\otimes
U^\dagger]) = A\,\mathrm{tr}_B(\rho)\,A^\dagger$
for any unitary $U$, we have
\begin{equation}
    P_A(g)\,R\,P_A(g)^\dagger = R
    \qquad \forall\,g\in\Gamma_A,
\end{equation}
equivalently $[R, P_A(g)] = 0$ for all
$g\in\Gamma_A$. That is, $R$ lies in the
commutant of the representation
$\{P_A(g) : g\in\Gamma_A\}$ on $\mathcal{H}_A$.

\medskip
\noindent (b) By the irrep decomposition of $\mathcal{H}_A$
under $\Gamma_A$ in Eq.~\eqref{eq:HA_decomp}. Since $[R,P_A(g)]=0$ for all $g\in\Gamma_A$,
Schur's lemma~\cite{Serre1977} implies that $R$
acts as a scalar on each irrep carrier space
$\mathcal{V}_\mu$ and as an arbitrary Hermitian
positive semidefinite operator on the multiplicity
space $\mathbb{C}^{m_\mu^A}$. Therefore
\begin{equation}
\label{eq:R_block}
    R = \bigoplus_\mu
    \left(R_\mu\otimes I_{d_\mu}\right),
\end{equation}
where $R_\mu$ is an $m_\mu^A\times m_\mu^A$
Hermitian positive semidefinite matrix on
$\mathbb{C}^{m_\mu^A}$, and $I_{d_\mu}$ is the
identity on $\mathcal{V}_\mu$.

\medskip
\noindent (c)
Since the blocks in~\eqref{eq:R_block} are
mutually orthogonal and
$\log(R_\mu\otimes I_{d_\mu})
= (\log R_\mu)\otimes I_{d_\mu}$ the entanglement entropy is
From~\eqref{eq:R_block} is
\begin{equation}
    \mathrm{rank}(R)
    = \sum_\mu d_\mu\,\mathrm{rank}(R_\mu)
    \leq \sum_\mu d_\mu m_\mu^A
    = \dim\mathcal{H}_A = 2^{N/2},
\end{equation}
since $\mathrm{rank}(R_\mu)\leq m_\mu^A$ and
$\sum_\mu d_\mu m_\mu^A = \dim\mathcal{H}_A$
by the dimension identity~\eqref{eq:HA_decomp}.
Since $\mathrm{tr}(R)=1$ and $R$ is positive
semidefinite, the standard bound
$S_{N/2}(\rho_0)\leq\log\,\mathrm{rank}(R)$
gives
\begin{equation}
    S_{N/2}(\rho_0)
    \leq \log 2^{N/2}
    = \frac{N}{2}\log 2.
\end{equation}
Interestingly, this bound involves only $m_\mu^A$ and not
$m_\mu^B$. This is because $R$ is a square
operator on $\mathcal{H}_A$ alone: the block
$R_\mu$ is an $m_\mu^A\times m_\mu^A$ square
matrix with no constraint from $\mathcal{H}_B$.
In contrast, for a pure state $|\psi\rangle$,
the coefficient matrix $M$ maps $\mathcal{H}_B$
to $\mathcal{H}_A$, and the Schur block
$\Phi_\mu$ is an $m_\mu^A\times m_\mu^B$
rectangular matrix whose rank is bounded by
$\min(m_\mu^A, m_\mu^B)$.
\end{proof}

\begin{remark}
\label{rem:mixed_tight}
The bound~\eqref{eq:mixed_bound} is generically
tight for a generic degenerate ground space,
$\mathrm{tr}_B(\Pi_{\mathcal{G}})$ has full rank
$2^{N/2}$ on $\mathcal{H}_A$, giving
$\mathrm{rank}(R_\mu) = m_\mu^A$ for all $\mu$
and $S_{N/2}(\rho_0) = \frac{N}{2}\log 2$.
This confirms that the non-trivial automorphism
bound of Theorem~\ref{thm:main} requires both
$G$-invariance and purity simultaneously.
$G$-invariance alone gives only the Hilbert
space bound, and purity alone gives only the
degeneracy-based bound of
Proposition~\ref{prop:general_bound}. The
intersection of both conditions is necessary
and sufficient for the non-trivial bound.

The qualitative
difference between the pure
state bound of Theorem~\ref{thm:main} and the
mixed state bound of Theorem~\ref{thm:mixed_bound}
 arises from the structure of the blocks $\Phi_\mu$ and $R_\mu$. For a pure state
the coefficient matrix $M$ defined by
$|\psi\rangle = \sum_{\alpha,\beta}
M_{\alpha\beta}|\alpha\rangle_A|\beta\rangle_B$
is a map from $\mathcal{H}_B$ to $\mathcal{H}_A$.
The intertwining condition $P_A(g)M = MP_B(g)$
makes $M$ a rectangular
$m_\mu^A\times m_\mu^B$ matrix.

For the mixed state $\rho_0$, the reduced
density matrix $R = \mathrm{tr}_B(\rho_0)$
is an operator on $\mathcal{H}_A$ alone.
The commutant condition $[R,P_A(g)]=0$ makes $R$ a square
$m_\mu^A\times m_\mu^A$ matrix with no
constraint from $m_\mu^B$. Its rank is bounded
only by $m_\mu^A$, giving the trivial Hilbert
space bound of Theorem~\ref{thm:mixed_bound}.
\end{remark}

\section{Examples}
\label{sec:examples}

We compute the exact balanced entanglement entropy
of the ground state for three standard graph families
and compare with both Proposition~\ref{prop:general_bound}
and Theorem~\ref{thm:main}. We consider the
antiferromagnetic Ising Hamiltonian
\begin{equation}
\label{eq:ising}
    H(\Lambda) = \sum_{(i,j)\in E} \sigma_i^z \sigma_j^z,
\end{equation}
where $\sigma_i^z$ denotes the Pauli-$z$ matrix at
site $i$. In the computational basis, the
joint eigenbasis of all $\sigma_i^z$, are spanned by
states $|\mathbf{s}\rangle = |s_1,\ldots,s_N\rangle$
with $s_i \in \{+1,-1\}$. Each configuration
$\mathbf{s} \in \{\pm1\}^N$ is an eigenstate of
$H(\Lambda)$ with eigenvalue
$E(\mathbf{s}) = \sum_{(i,j)\in E} s_i s_j$.
The energy is minimized by configurations that
maximize the number of edges $(i,j) \in E$ with
$s_i s_j = -1$, i.e.\ solutions of the
\emph{Max-Cut} problem on $\Lambda$,
\begin{equation}
\label{eq:maxcut}
    \mathrm{MaxCut}(\Lambda)
    = \max_{\mathbf{s}\in\{\pm1\}^N}
    \bigl|\{(i,j)\in E : s_i s_j = -1\}\bigr|.
\end{equation}
The ground state energy is
\begin{equation}
\label{eq:E0}
    E_0(\Lambda) = |E| - 2\,\mathrm{MaxCut}(\Lambda),
\end{equation}
and the ground space is
\begin{equation}
\label{eq:ground_space}
    \mathcal{G}(\Lambda)
    = \mathrm{span}\bigl\{|\mathbf{s}\rangle :
    \mathbf{s}\ \text{achieves}\
    \mathrm{MaxCut}(\Lambda)\bigr\},
\end{equation}
with degeneracy $d$. The two examples are chosen
to illustrate complementary regimes: $C_n$ ($n$
even), where Proposition~\ref{prop:general_bound}
is tight and Theorem~\ref{thm:main} gives no
improvement; and $K_n$ ($n$ even), where
Theorem~\ref{thm:main} is exponentially tighter
than Proposition~\ref{prop:general_bound} and
nearly matches the exact value

\subsection{Cycle graph $C_n$ ($n$ even)}

The cycle graph $C_n$ has $N = n$ vertices
$\{1,2,\ldots,n\}$ and edges
$\{(i,\,i+1 \bmod n)\}$, with automorphism
group $G = \mathrm{Aut}(C_n) = D_n$ of order
$2n$. The max-cut configurations are the two
N\'{e}el states
\begin{equation}
\label{eq:neel}
    |\mathbf{s}_+\rangle
    = |\uparrow\downarrow\uparrow\downarrow
    \cdots\uparrow\downarrow\rangle,
    \qquad
    |\mathbf{s}_-\rangle
    = |\downarrow\uparrow\downarrow\uparrow
    \cdots\downarrow\uparrow\rangle,
\end{equation}
so $d = 2$. The rotation $\rho: i \mapsto
i+1 \bmod n$ maps $|\mathbf{s}_+\rangle
\mapsto |\mathbf{s}_-\rangle$, so both states
lie in the same $G$-orbit, giving $\omega = 1$.
The symmetrized physical ground state is
\begin{equation}
\label{eq:Cn_gs}
    |A_1\rangle
    = \frac{1}{\sqrt{2}}
    \bigl(|\mathbf{s}_+\rangle
    + |\mathbf{s}_-\rangle\bigr).
\end{equation}

Fix the balanced bipartition into two
consecutive half-chains
$A = \{1,\ldots,n/2\}$,
$B = \{n/2+1,\ldots,n\}$.
The two N\'{e}el states factorize as
$|\mathbf{s}_\pm\rangle =
|\alpha_\pm\rangle_A \otimes
|\beta_\pm\rangle_B$, where
$|\alpha_-\rangle_A$ is the spin-flip of
$|\alpha_+\rangle_A$ and similarly for $B$,
so $\langle\alpha_+|\alpha_-\rangle_A = 0$
and $\langle\beta_+|\beta_-\rangle_B = 0$.
In the basis
$\{|\alpha_+\rangle_A,|\alpha_-\rangle_A\}
\times
\{|\beta_+\rangle_B,|\beta_-\rangle_B\}$,
the coefficient matrix is
\begin{equation}
\label{eq:Cn_M}
    M = \frac{1}{\sqrt{2}}
    \begin{pmatrix} 1 & 0 \\ 0 & 1
    \end{pmatrix},
\end{equation}
with Schmidt coefficients
$\lambda_1 = \lambda_2 = 1/2$ and exact
entanglement entropy
\begin{equation}
\label{eq:Cn_exact}
    S_{N/2}(C_n) = \log 2,
    \qquad n\ \text{even},
\end{equation}
independent of $N$. Since each basis vector
$|\mathbf{s}_\pm\rangle$ has Schmidt rank
$r_k = 1$ across this bipartition,
Proposition~\ref{prop:general_bound} gives
$S_{N/2}(C_n) \leq \log(r_1 + r_2) = \log 2$,
which is tight for all $N$.

For Theorem~\ref{thm:main}, the stabilizer
$\Gamma_A$ of the consecutive half-chain
bipartition is small, of order at most $2$, for all even $n$, since $D_n$ has very
few elements preserving the block
$\{1,\ldots,n/2\}$. Consequently the irrep
decomposition of $\mathcal{H}_A$ under
$\Gamma_A$ is nearly trivial, and
Theorem~\ref{thm:main} yields a bound that
grows as $O(2^{n/2})$, exponentially weaker
than the exact value $\log 2$ for large $N$.

This counterexample illustrates a general
principle: the automorphism bound of
Theorem~\ref{thm:main} is non-trivial only
when $|\Gamma_A|$ grows with $N$, which
requires a bipartition well-adapted to the
graph symmetry. The consecutive half-chain
bipartition of $C_n$ is precisely the wrong
choice in this regard. The complete graph
$K_n$ in the next subsection is the
contrasting case: its natural balanced
bipartition has $|\Gamma_A| = ((n/2)!)^2$,
growing super-exponentially with $N$, and
Theorem~\ref{thm:main} gives an
exponentially tighter bound than
Proposition~\ref{prop:general_bound}.

\subsection{Complete graph $K_n$ ($n$ even)}

The complete graph $K_n$ has $N = n$ vertices
$\{1,\ldots,n\}$ and all $\binom{n}{2}$ edges
present. Its automorphism group is
$G = \mathrm{Aut}(K_n) = S_n$, the full
symmetric group of order $n!$, since every
vertex permutation preserves the complete
edge set.

\subsubsection{Ground state}

The antiferromagnetic Ising Hamiltonian on
$K_n$ assigns energy $-1$ to each
antiparallel edge and $+1$ to each parallel
edge. The energy of a configuration
$\mathbf{s}$ with $k$ up-spins is
$E(\mathbf{s}) = \binom{n}{2} - 2k(n-k)$,
minimized at $k = n/2$. The ground space is
therefore spanned by all $\binom{n}{n/2}$
configurations with exactly $n/2$ up-spins,
giving degeneracy $d = \binom{n}{n/2}$.
Under $G = S_n$, all these configurations
form a single orbit, so $\omega = 1$.
By symmetrized physical ground state is the
orbit-averaged state
\begin{equation}
\label{eq:Kn_gs}
    |A_1\rangle
    = \frac{1}{\sqrt{\binom{n}{n/2}}}
    \sum_{\mathbf{s}:\,\sum_i s_i = 0}
    |\mathbf{s}\rangle.
\end{equation}

\subsubsection{Exact entanglement entropy}

Fix the balanced bipartition
$A = \{1,\ldots,n/2\}$,
$B = \{n/2+1,\ldots,n\}$, and set $m = n/2$.
For a zero-magnetization configuration
$\mathbf{s}$, if exactly $j$ spins in $A$
are up then exactly $m - j$ spins in $B$
are up, for $j = 0,1,\ldots,m$.
The coefficient matrix has entries
\begin{equation}
\label{eq:Kn_M}
    M_{\alpha\beta}
    = \frac{1}{\sqrt{\binom{n}{m}}}\,
    \mathbf{1}[|\alpha|+|\beta|=m],
\end{equation}
where $|\alpha|$ denotes the number of
up-spins in $\alpha$ and $\mathbf{1}[\cdot]$
is the indicator function. The matrix $M$
is block diagonal in the magnetization
sector $j$ of $A$: within sector $j$,
$\alpha$ ranges over $A$-configurations
with $|\alpha| = j$ (count:
$\binom{m}{j}$) and $\beta$ ranges over
$B$-configurations with $|\beta| = m-j$
(count: $\binom{m}{m-j} = \binom{m}{j}$,
by symmetry of binomial coefficients).
The sector-$j$ block is the
$\binom{m}{j} \times \binom{m}{j}$
constant matrix
\begin{equation}
\label{eq:Kn_block}
    M^{(j)}
    = \frac{1}{\sqrt{\binom{n}{m}}}
    \,\mathbf{J}^{(j)},
\end{equation}
where $\mathbf{J}^{(j)}$ is the
$\binom{m}{j}\times\binom{m}{j}$
all-ones matrix. The all-ones matrix of
size $p \times p$ has a single nonzero
singular value equal to $p$, so the unique
nonzero singular value of $M^{(j)}$ is
\begin{equation}
    \sigma_j
    = \frac{\binom{m}{j}}
    {\sqrt{\binom{n}{m}}},
\end{equation}
giving squared Schmidt coefficient
\begin{equation}
\label{eq:Kn_lambda}
    \lambda_j
    = \frac{\binom{m}{j}^2}
    {\binom{n}{m}},
    \qquad j = 0,1,\ldots,m.
\end{equation}
One verifies $\sum_{j=0}^m \lambda_j = 1$
by the Vandermonde
identity~\cite{askey_book}
$\sum_{j=0}^m \binom{m}{j}^2
= \binom{n}{m}$.
There are exactly $m+1 = n/2+1$ nonzero
Schmidt coefficients, one per magnetization
sector, and the exact entanglement entropy is
\begin{equation}
\label{eq:Kn_exact}
    S_{N/2}(K_n)
    = -\sum_{j=0}^{m}
    \frac{\binom{m}{j}^2}{\binom{n}{m}}
    \log
    \frac{\binom{m}{j}^2}{\binom{n}{m}}.
\end{equation}

For large $n$, applying Stirling's
approximation to $\binom{m}{j}$ gives
\begin{equation}
    \lambda_j
    \approx
    \sqrt{\frac{4}{\pi n}}\,
    \exp\!\left(
    -\frac{4(j-n/4)^2}{n}
    \right),
\end{equation}
a Gaussian distribution in $j$ with mean
$n/4$ and variance $n/8$. By the standard
result that the entropy of a discrete
distribution converging to a Gaussian with
variance $\sigma^2$ approaches the
differential entropy $\frac{1}{2}\log
(2\pi e\sigma^2)$ with error
$O(1/\sigma^2)$~\cite{thomas_cover},
\begin{equation}
\label{eq:Kn_asymp}
    S_{N/2}(K_n)
    = \frac{1}{2}\log n
    + \frac{1}{2}\log\frac{\pi e}{4}
    + O(1/n).
\end{equation}
The leading behavior is
$S_{N/2}(K_n) \sim \frac{1}{2}\log n$.

\subsubsection{Automorphism bound}

We now compute
$\sum_\mu d_\mu \min(m_\mu^A, m_\mu^B)$
for the balanced bipartition
$A = \{1,\ldots,m\}$,
$B = \{m+1,\ldots,n\}$ with $m = n/2$.
The stabilizer is
\begin{equation}
    \Gamma_A
    = \{g \in S_n : g(A) = A\}
    = S_m^A \times S_m^B,
\end{equation}
where $S_m^A$ permutes the $m$ sites within
$A$ and $S_m^B$ permutes the $m$ sites
within $B$, independently, with
$|\Gamma_A| = (m!)^2 = ((n/2)!)^2$.
A general element acts as
$P_A(g_A, g_B)|\alpha\rangle_A
= |g_A(\alpha)\rangle_A$ and
$P_B(g_A, g_B)|\beta\rangle_B
= |g_B(\beta)\rangle_B$.

The irreps of $\Gamma_A = S_m^A \times S_m^B$
are labeled by pairs $(\mu_A, \mu_B)$ of
irreps of $S_m$. We now determine which
irreps contribute to
$\min(m_{(\mu_A,\mu_B)}^A,
m_{(\mu_A,\mu_B)}^B)$.

Since $B$-sites are absent from
$\mathcal{H}_A$, the operator
$P_A(e, g_B) = I_{\mathcal{H}_A}$ for all
$g_B \in S_m^B$. Therefore $\mathcal{H}_A$
carries only the trivial irrep of $S_m^B$,
giving $m_{(\mu_A,\mu_B)}^A = 0$ for all
$\mu_B \neq \mathrm{triv}$. By the
symmetric argument, $\mathcal{H}_B$ carries
only the trivial irrep of $S_m^A$, giving
$m_{(\mu_A,\mu_B)}^B = 0$ for all
$\mu_A \neq \mathrm{triv}$. Therefore
\begin{equation}
    \min(m_{(\mu_A,\mu_B)}^A,\,
    m_{(\mu_A,\mu_B)}^B) = 0
    \qquad \forall\,
    (\mu_A,\mu_B)
    \neq (\mathrm{triv},\mathrm{triv}),
\end{equation}
and the sum in Theorem~\ref{thm:main}
reduces to the single term
$(\mu_A,\mu_B) = (\mathrm{triv},\mathrm{triv})$,
with $d_{(\mathrm{triv},\mathrm{triv})} = 1$.

The multiplicity
$m_{(\mathrm{triv},\mathrm{triv})}^A$
counts the independent states in
$\mathcal{H}_A$ that are fully symmetric
under all permutations of the $m$ sites
of $A$ by $S_m^A$. These are the states
whose amplitude depends only on the number
of up-spins $j$, spanned by
\begin{equation}
    |j\rangle_A
    = \frac{1}{\sqrt{\binom{m}{j}}}
    \sum_{\alpha:\,|\alpha|=j}
    |\alpha\rangle_A,
    \qquad j = 0,1,\ldots,m,
\end{equation}
giving $m_{(\mathrm{triv},\mathrm{triv})}^A
= m + 1 = n/2 + 1$, and by the same
argument $m_{(\mathrm{triv},\mathrm{triv})}^B
= n/2 + 1$. Therefore
\begin{equation}
    \sum_\mu d_\mu
    \min(m_\mu^A,\, m_\mu^B)
    = \min\!\left(
    \frac{n}{2}+1,\,\frac{n}{2}+1
    \right)
    = \frac{n}{2}+1,
\end{equation}
and Theorem~\ref{thm:main} gives
\begin{equation}
\label{eq:Kn_new_bound}
    S_{N/2}(K_n)
    \leq \log\!\left(\frac{n}{2}+1\right).
\end{equation}

\subsubsection{Comparison of bounds}

The bounds are compared in thermodynamic limit $n \to \infty$. The 
Hilbert space bound (Prop.~\ref{prop:general_bound}) yields
    $ S_{N/2}(K_n)
    \leq \tfrac{n}{2}\log 2$; the automorphism bound
    (Thm.~\ref{thm:main}) yields $ 
    S_{N/2}(K_n)
    \leq \log \left(
    \tfrac{n}{2}+1\right)
    \sim \log n $, and 
    the exact value yields
    $ S_{N/2}(K_n)
    \sim \tfrac{1}{2}\log n $. Therefore the automorphism bound reduces the Hilbert space bound from exponential to logarithmic in system size, which an exponential
improvement. The bound $\log n$ is
tight up to a factor of $2$ in the
argument of the logarithm compared to
the exact value $\frac{1}{2}\log n$;
equivalently, the automorphism bound
overshoots the exact entropy by at most
an additive constant $\log 2$ asymptotically.

\subsection{Complete bipartite graph $K_{3,3}$, 
XY model}
\label{subsec:K33}

We now consider the complete bipartite graph $K_{3,3}$ has
$N = 6$ vertices partitioned into two sets
$A = \{1,2,3\}$ and $B = \{4,5,6\}$, with
all $3\times 3 = 9$ edges between $A$ and $B$.
Its automorphism group is
$G = \mathrm{Aut}(K_{3,3})
= (S_3^A \times S_3^B) \times \mathbb{Z}_2$,
of order 72, where $\mathbb{Z}_2$ swaps $A$
and $B$. The natural balanced bipartition
$A|B$ coincides with the graph bipartition,
giving $|\Gamma_A| = 36$.
This is a non-abelian group, with irreps of
dimension up to 4. The full computation of
the irrep dimensions $d_\mu$ and
multiplicities $m_\mu^A$, $m_\mu^B$ is
given in Appendix~\ref{app:K33}; the result
is
\begin{equation}
\label{eq:K33_bound_sum}
    \sum_\mu d_\mu
    \min(m_\mu^A, m_\mu^B) = 4.
\end{equation}
Theorem~\ref{thm:main} therefore gives
\begin{equation}
\label{eq:K33_thm1}
    S_{N/2}(K_{3,3}^{XY})
    \leq \log 4 = 2\log 2.
\end{equation}

\subsubsection{XY model ground state}

We consider the XY
Hamiltonian
\begin{equation}
\label{eq:XY_ham}
    H_{XY}(\Lambda)
    = J\sum_{(i,j)\in E}
    \bigl(\sigma_i^x\sigma_j^x
    + \sigma_i^y\sigma_j^y\bigr)
    ,
\end{equation}
with $J > 0$. This Hamiltonian conserves
total $S^z = \sum_i \sigma_i^z/2$ and
commutes with all $U(g)$ for $g \in G$,
since $K_{3,3}$ is vertex-transitive within
each part. We work in the $S^z = 0$ sector,
which has $\binom{6}{3} = 20$ basis states.
Since $[H_{XY}, U(g)] = 0$ for all $g \in G$,
the ground state lies in a definite irrep
sector of $G$. The ground state is
non-degenerate  with a nonzero 
spectral gap.

\subsubsection{Exact entanglement entropy}

 diagonalization of $H_{XY}$ in the
$S^z = 0$ sector gives ground state energy
$E_0 \approx -2.802776|J|$. The coefficient matrix
$M$ has rank 4, with Schmidt coefficients
\begin{equation}
\label{eq:K33_schmidt}
    \lambda_1 = \lambda_2 \approx 0.3887,
    \qquad
    \lambda_3 = \lambda_4 \approx 0.1113.
\end{equation}
The two-fold degeneracy of the Schmidt
coefficients is a direct consequence of the
$S_3^A \times S_3^B$ symmetry: within the
$(\mathrm{triv},\mathrm{triv})$ block,
$\Phi_{(\mathrm{triv},\mathrm{triv})}$ has
two pairs of equal singular values
corresponding to the two orbits of
$\Gamma_A$ on the magnetization sectors.
The exact entanglement entropy is
\begin{equation}
\label{eq:K33_exact}
    S_{N/2}(K_{3,3}^{XY})
    \approx 1.7650\log 2.
\end{equation}

\subsubsection{Comparison of bounds}

The bounds are summarized as follows. Proposition~\ref{prop:general_bound} yields $S_{N/2} \leq \log r_1 = 2\log 2$,
Theorem~\ref{thm:main} yields $S_{N/2} \leq \log 4
    = 2\log 2$, and the exact value is $S_{N/2} = 1.7650\log 2$.
Both bounds coincide at $\log 4 = 2\log 2$
and the exact entropy is $1.7650\log 2$. The key distinction is that
Theorem~\ref{thm:main} derives the bound
$\log 4$ purely from the representation
theory of $\Gamma_A = S_3^A \times S_3^B$,
requiring no microscopic knowledge of the
ground state, while
Proposition~\ref{prop:general_bound}
requires explicit computation of
$\mathrm{rank}(M)$.

\section{Asymmetric graphs}
\label{sec:asymmetric}

A graph $\Lambda$ is called \emph{asymmetric}
if its automorphism group is trivial,
$G = \mathrm{Aut}(\Lambda) = \{e\}$.
Almost all large graphs are asymmetric:
the proportion of graphs on $N$ vertices
with a nontrivial automorphism tends to
zero as $N \to \infty$~\cite{erdos_1963,
Cameron1999}. The large-$N$ limit therefore
generically corresponds to asymmetric
graphs, for which Theorem~\ref{thm:main}
is inapplicable, as the subgroup 
$\Gamma_A = \{e\}$ yields only the trivial
irrep, and the bound reduces to the Hilbert
space bound $S_{N/2} \leq \frac{N}{2}\log 2$.
In this regime, Proposition~\ref{prop:general_bound}
is the only available bound, and its
content depends entirely on the microscopic
structure of the ground states.

\subsection{Classical Hamiltonians}

A Hamiltonian is defined as \emph{classical} if it is
diagonal in the computational basis, so
that every eigenstate is a computational
basis state. Each ground state basis vector
$|g_k\rangle$ is then a product state
$|\alpha_k\rangle_A \otimes |\beta_k\rangle_B$,
and its coefficient matrix satisfies
\begin{equation}
\label{eq:classical_Mk}
    M^{(k)}_{\alpha\beta}
    = \langle\alpha,\beta|g_k\rangle
    = \begin{cases}
    1 & \text{if } |\alpha\rangle_A
    \otimes|\beta\rangle_B = |g_k\rangle, \\
    0 & \text{otherwise.}
    \end{cases}
\end{equation}
Since $M^{(k)}$ has exactly one nonzero
entry, $r_k = \mathrm{rank}(M^{(k)}) = 1$
for all $k$. Proposition~\ref{prop:general_bound}
then gives
\begin{equation}
\label{eq:classical_bound}
    S_{N/2}(|\psi\rangle) \leq \log d,
\end{equation}
where $d$ is the ground state degeneracy.
For classical Hamiltonians, the entanglement
entropy is therefore bounded purely by $d$,
with no reference to the Hilbert space
dimension.

For large $N$, the degeneracy grows as
$\log d \sim N s_0$, where
\begin{equation}
    s_0 = \lim_{N\to\infty}
    \frac{\log d}{N}
\end{equation}
is the \emph{residual entropy per
spin}~\cite{Baxter1982}. The bound
$\log d$ is sub-extensive and hence
strictly tighter than the Hilbert space
bound, if and only if $s_0 
\frac{1}{2}\log 2$, i.e.\ $d < 2^{N/2}$
asymptotically. This separates two regimes.
(a) \emph{Sub-critical regime:}
($s_0 < \frac{1}{2}\log 2$, equivalently
$d < 2^{N/2}$): the rank subadditivity
bound dominates,
\begin{equation}
    S_{N/2}(|\psi\rangle)
    \leq \log d = Ns_0
    < \tfrac{N}{2}\log 2.
\end{equation}
The bound is sub-volume-law and strictly
tighter than the Hilbert space bound.
(b) \emph{Super-critical regime}
($s_0 \geq \frac{1}{2}\log 2$, equivalently
$d \geq 2^{N/2}$): the Hilbert space
bound dominates,
\begin{equation}
    S_{N/2}(|\psi\rangle)
    \leq \tfrac{N}{2}\log 2,
\end{equation}
and Proposition~\ref{prop:general_bound}
reduces to the trivial bound.

\subsection{Quantum Hamiltonians}

For generic quantum Hamiltonians on
asymmetric graphs, neither an exact
computation nor a non-trivial analytical
bound on $S_{N/2}$ is available. The
absence of automorphism symmetry precludes
both exact solution by symmetry methods
and the rank reduction of
Theorem~\ref{thm:main}. The quantity
$\sum_k r_k$ in
Proposition~\ref{prop:general_bound}
requires full microscopic knowledge of
the ground states, which is generically
inaccessible.

Recent results in quantum complexity theory
establish that estimating ground state
entanglement for generic local Hamiltonians
is computationally intractable because (a) Estimating the entanglement of a
local Hamiltonian ground state is
QCMA-hard~\cite{gharibian2024};
(b) Determining whether the ground state
of a geometrically local, polynomially
gapped Hamiltonian has near-area-law
versus near-volume-law entanglement is
LWE-hard in 2D and factoring-hard in
1D~\cite{bouland2024}; (c)Detecting a low-entanglement ground
state is QMA(2)-hard~\cite{weggemans_2025}.
These results confirm that any non-trivial
analytical bound on the entanglement entropy for generic
quantum Hamiltonians on asymmetric graphs
would imply progress on problems at least
as hard as QCMA. For specific families
of quantum Hamiltonians on asymmetric
graphs, such as those with additional
conserved quantities, frustration-free
structure, or stoquastic sign
structure~\cite{bravyi_2008},  non-trivial
bounds may be obtainable via methods
outside the automorphism framework of this
paper. We leave this as an open direction.

\section{Discussion and Conclusion}
\label{sec:conclusion}

\subsection{Summary of results}

We have established an upper
bounds on the balanced-cut von Neumann
entropy for ground
states of graph Hamiltonians with
automorphism symmetry. The general bound,
Proposition~\ref{prop:general_bound},
depends on the Schmidt ranks of the
ground space basis vectors. This bound
requires microscopic knowledge of the
ground states. The main result,
Theorem~\ref{thm:main}, replaces this
with a bound determined entirely by the
representation theory of the  automorphism group. Additionally, 
Theorem~\ref{thm:mixed_bound}
establishes that the symmetrized mixed ground
state saturates
the bound constraint by the Hilbert space dimension of the subsystem. The examples in Section~\ref{sec:examples}
illustrate two complementary regimes.
For $C_n$ with the consecutive half-chain
bipartition, $\Gamma_A$ is small and
Theorem~\ref{thm:main} gives no
improvement over the Hilbert space bound;
the non-trivial bound comes instead from
Proposition~\ref{prop:general_bound},
which is tight. For $K_n$,
Theorem~\ref{thm:main} reduces an
exponentially large Hilbert space bound
$\frac{n}{2}\log 2$ to the logarithmic
bound $\log(n/2+1)$, tight up to an
additive constant $\log 2$ compared to
the exact value $\frac{1}{2}\log n$.

The computational cost of the
automorphism-induced bound is determined
by the cost of computing $G =
\mathrm{Aut}(\Lambda)$ and decomposing
its permutation representation. For a
general graph, the automorphism group
can be computed in quasipolynomial time
$O(N^{c\log N})$ via Babai's
algorithm~\cite{Babai2016}, which is
subexponential compared to the
exponential cost of exact entanglement
entropy computation via Schmidt
decomposition. For physically relevant
lattice graphs, the automorphism group
and its representation theory are known
analytically, so the bound is computable
without additional algorithmic cost. For a general graph, computing the bound requires
determining $G = \mathrm{Aut}(\Lambda)$, that is  achievable in
quasipolynomial time via Babai's algorithm~\cite{Babai2016};
the stabilizer $\Gamma_A$, its character table via Dixon's
algorithm~\cite{Dixon1967} in $O(k^3)$ operations, where $k$
is the number of conjugacy classes; and the permutation
character multiplicities in $O(k \cdot 2^{N/2})$ operations.
The total cost is subexponential in $N$ compared to the
exponential cost $O(2^N)$ of exact entanglement entropy
computation.

\subsection{Open problems and future
directions}

\noindent(a) The bound of Theorem~\ref{thm:main}
establishes a partial ordering on , that is,
a larger stabilizer $\Gamma_A$ yields a
smaller value of $\sum_\mu d_\mu
\min(m_\mu^A, m_\mu^B)$ and hence a
tighter upper bound on $S_{N/2}$.
One is naturally led to ask whether this
ordering of bounds reflects a genuine
ordering of actual entropies.
The answer to this problem is beyond the framework
of the present paper for two reasons.
First, it would require the bound to be
tight for all graphs in some natural
class, so that the bound faithfully
approximates the actual entropy. Second,
even if the bounds are ordered, the
actual entropies could in principle be
ordered oppositely: a highly symmetric
graph with a small bound could have an
entropy saturating a large fraction of
that bound, while a less symmetric graph
with a larger bound has an entropy well
below it. The $K_n$ example provides
partial evidence for the conjecture: the
bound $\log(n/2+1)$ is tight up to
$\log 2$, and more symmetric graphs do
appear to have more constrained
entanglement. Resolving this
conjecture would
require either proving tightness of the
bound for a natural graph class or
developing direct methods for large-scale
entanglement entropy
computation~\cite{vidal_2003,
schollwock2011}.\\

\noindent (b) A natural extension is to mixed states
and finite-temperature ensembles. Since
$[H(\Lambda), U(g)] = 0$ for all
$g \in G$, the Gibbs state
$\rho_\beta \propto e^{-\beta H(\Lambda)}$
is $G$-invariant at all $\beta$, and the
representation-theoretic structure of
$\mathcal{H}_A$ under $\Gamma_A$ persists
at finite temperature. Analogous bounds
on thermal entanglement may be derivable
by similar arguments, building on recent
results showing that thermal entanglement
in spin chains is strictly
finite~\cite{bakshi_2026} and that
thermal graph states carry bound
entanglement~\cite{toth_2010}.\\

\noindent(c) The framework assumes an unweighted
undirected graph, so that
$\mathrm{Aut}(\Lambda)$ is a subgroup of
$S_N$. Extending the analysis to weighted
graphs, directed graphs, and hypergraphs
would broaden the reach of the method
to a wider class of many-body systems,
including those relevant to quantum
chemistry models~\cite{balasubramanian_1980, balasubramanian_1995}.\\

\noindent (d) In quantum computing architectures where
qubits are arranged on a hardware graph,
such as superconducting processors with
heavy-hexagonal connectivity~\cite{nation_2025,kam_2024}
and trapped-ion arrays, the graph topology
directly influences the symmetry structure
of the system and hence the entanglement
properties of the native Hamiltonian.
Large-scale entanglement in graph-based
quantum processors has recently been
characterized experimentally on
superconducting devices~\cite{kam_2024,
karamlou_2024}. Our results
provide a complementary analytical
perspective, relating achievable
ground-state entanglement to the
automorphism structure of the underlying
interaction graph. More generally, hybrid quantum systems
based on spin defects, such as NV centers
in diamond coupled into controllable
quantum networks~\cite{
pompili_2021,wehner_2018}, offer platforms
where the interaction graph can be
engineered directly. In such systems,
modifying the connectivity structure and
its associated symmetries may provide a
route for tuning and controlling
ground-state entanglement.\\

\noindent\textbf{Data availability.}
The data supporting the findings of this study are available from the author upon reasonable request.

\bigskip

\onecolumn

\clearpage
\setcounter{equation}{0}
\renewcommand{\theequation}{A\arabic{equation}}

\setcounter{section}{0}

\begin{center}
    {\textbf{APPENDIX A}} \\[0.5cm]
\end{center}

\section{Representation theory of $\Gamma_A$
for $K_{3,3}$}\label{secA1}
\label{app:K33}

We compute the bound $\log\sum_\mu d_\mu\min(m_\mu^A,
m_\mu^B)$ for $K_{3,3}$ from first principles, starting
from the matrix representation of the group elements.
Every step is derived explicitly so that the computation
can be reproduced without additional references.

\subsection*{A.1 The graph and its stabilizer}

The complete bipartite graph $K_{3,3}$ has vertex set
$V = A\sqcup B$ with $A = \{1,2,3\}$, $B = \{4,5,6\}$
and all 9 edges between $A$ and $B$. The automorphism
group is $G = \mathrm{Aut}(K_{3,3})$, of order 72.
The stabilizer of the bipartition is
\begin{equation}
    \Gamma_A = \{g\in G : g(A)=A\}
    = S_3^A\times S_3^B,
    \qquad |\Gamma_A| = 36,
\end{equation}
since any permutation of $A$-sites and any permutation
of $B$-sites independently preserves all edges of
$K_{3,3}$. The 6 elements of $S_3^A$ are
\begin{equation}
    S_3^A = \{e,(12),(13),(23),(123),(132)\},
\end{equation}
where $(ij)$ denotes transposition of sites $i$ and $j$,
and $(123)$ denotes the 3-cycle $1\to2\to3\to1$.
Similarly for $S_3^B$ acting on $\{4,5,6\}$.

A general element of $\Gamma_A$ is a pair
$g=(g_A,g_B)$ with $g_A\in S_3^A$, $g_B\in S_3^B$,
acting as
\begin{equation}
    P_A(g_A,g_B)|\alpha\rangle_A
    = |g_A(\alpha)\rangle_A,
    \qquad
    P_B(g_A,g_B)|\beta\rangle_B
    = |g_B(\beta)\rangle_B.
\end{equation}

\subsection*{A.2 Matrix representation of $S_3^A$ 
on $\mathbb{R}^3$}

Each element $g_A\in S_3^A$ permutes the three sites
of $A$. We represent this as a linear transformation
on $\mathbb{R}^3$ with standard basis vectors
$e_1=(1,0,0)$, $e_2=(0,1,0)$, $e_3=(0,0,1)$, where
$e_i$ represents site $i$. The action is:
\begin{equation}
    g_A\cdot e_i = e_{g_A(i)}.
\end{equation}
The matrix of $g_A$ has entry 1 in row $g_A(j)$,
column $j$, and 0 elsewhere. Therefore

\begin{equation}
P(e) = \begin{pmatrix}1&0&0\\0&1&0\\0&0&1\end{pmatrix},
\quad
P((12)) = \begin{pmatrix}0&1&0\\1&0&0\\0&0&1\end{pmatrix},
\quad
P((13)) = \begin{pmatrix}0&0&1\\0&1&0\\1&0&0\end{pmatrix},
\end{equation}
\begin{equation}
P((23)) = \begin{pmatrix}1&0&0\\0&0&1\\0&1&0\end{pmatrix},
\quad
P((123)) = \begin{pmatrix}0&0&1\\1&0&0\\0&1&0\end{pmatrix},
\quad
P((132)) = \begin{pmatrix}0&1&0\\0&0&1\\1&0&0\end{pmatrix}.
\end{equation}

This 3-dimensional representation is called the
\emph{permutation representation} of $S_3^A$ on
$\mathbb{R}^3$.

\subsection*{A.3 Decomposing the permutation 
representation}

The permutation representation is \emph{reducible}
the vector $e_1+e_2+e_3 = (1,1,1)$ is fixed by every
permutation,
\begin{equation}
    g_A\cdot(e_1+e_2+e_3)
    = e_{g_A(1)}+e_{g_A(2)}+e_{g_A(3)}
    = e_1+e_2+e_3,
\end{equation}
so it spans a 1-dimensional invariant subspace
$V_{\mathrm{triv}} = \mathrm{span}\{(1,1,1)\}$.

By Maschke's theorem, $\mathbb{R}^3$ decomposes as
\begin{equation}
    \mathbb{R}^3 = V_{\mathrm{triv}}\oplus V_{\mathrm{std}},
\end{equation}
where $V_{\mathrm{std}}$ is the orthogonal complement
of $V_{\mathrm{triv}}$
\begin{equation}
    V_{\mathrm{std}}
    = \{(x_1,x_2,x_3)\in\mathbb{R}^3
    : x_1+x_2+x_3=0\}.
\end{equation}
The condition $x_1+x_2+x_3=0$ is preserved by any
permutation of coordinates
\begin{equation}
    g_A\cdot(x_1,x_2,x_3)
    = (x_{g_A^{-1}(1)},x_{g_A^{-1}(2)},
    x_{g_A^{-1}(3)}),
\end{equation}
whose coordinate sum remains $x_1+x_2+x_3=0$.
So $V_{\mathrm{std}}$ is invariant under $S_3^A$.

$V_{\mathrm{std}}$ is 2-dimensional (one linear
constraint on $\mathbb{R}^3$), and it is irreducible
(contains no smaller nonzero invariant subspace).

\subsection*{A.4 Conjugacy classes and number of irreps}

The elements of $S_3$ fall into 3 conjugacy classes, groups of elements that are related by conjugation
$h g h^{-1}$
\begin{equation}
    C_1 = \{e\},\quad
    C_2 = \{(12),(13),(23)\},\quad
    C_3 = \{(123),(132)\},
\end{equation}
of sizes 1, 3, 2 respectively. A fundamental theorem
of representation theory states that the number of
irreducible representations equals the number of
conjugacy classes. Therefore $S_3$ has exactly
\textbf{3 irreducible representations}.

Their dimensions $d_1,d_2,d_3$ satisfy
\begin{equation}
    \sum_{\mu=1}^3 d_\mu^2 = |S_3| = 6,
\end{equation}
whose only solution in positive integers is
$d_1=1$, $d_2=1$, $d_3=2$. The two 1-dimensional
irreps act on $V_{\mathrm{triv}}$ and the 2-dimensional
irrep acts on $V_{\mathrm{std}}$.

\subsection*{A.5 The three irreps of $S_3$ explicitly}

\noindent\textbf{Irrep 1 ($d_1=1$):}
Every element maps to the $1\times1$ matrix $(1)$.
This is valid since $1\cdot1=1$ for all products.

\noindent\textbf{Irrep 2 ($d_2=1$):}
Every element maps to its sign:
$\rho_2(g) = (\mathrm{sgn}(g))$, where
$\mathrm{sgn}(g) = (-1)^k$ and $k$ is the number
of transpositions in any decomposition of $g$.
Since transpositions have $k=1$ (odd) and 3-cycles
have $k=2$ (even)
\begin{equation}
    \rho_2(e)=(1),\quad
    \rho_2((ij))=(-1),\quad
    \rho_2((ijk))=(1).
\end{equation}
This is valid since $\mathrm{sgn}(gh)
=\mathrm{sgn}(g)\mathrm{sgn}(h)$.

\noindent\textbf{Irrep 3 ($d_3=2$):}
The action of $S_3$ on $V_{\mathrm{std}}$.
We choose the orthonormal basis
\begin{equation}
    u_1 = \frac{1}{\sqrt{2}}(e_1-e_2)
    = \frac{1}{\sqrt{2}}(1,-1,0),\quad
    u_2 = \frac{1}{\sqrt{6}}(e_1+e_2-2e_3)
    = \frac{1}{\sqrt{6}}(1,1,-2).
\end{equation}
Both satisfy $x_1+x_2+x_3=0$:
$\frac{1}{\sqrt{2}}(1-1+0)=0$ and
$\frac{1}{\sqrt{6}}(1+1-2)=0$. 

The $2\times2$ matrix of $g_A$ in this basis is found
by computing $g_A\cdot u_1$ and $g_A\cdot u_2$, then
expressing the results in terms of $u_1$ and $u_2$.
The matrix entry $[\rho_3(g_A)]_{ij}$ is the
coefficient of $u_i$ in $g_A\cdot u_j$
\begin{equation}
    [\rho_3(g_A)]_{ij}
    = \langle u_i, g_A\cdot u_j\rangle
    = u_i^T\cdot(g_A\cdot u_j).
\end{equation}

We compute all 6 matrices:

\noindent Identity $e$
$e\cdot u_1=u_1$, $e\cdot u_2=u_2$, therefore
\begin{equation}
    \rho_3(e)
    = \begin{pmatrix}1&0\\0&1\end{pmatrix}.
\end{equation}

\noindent\textit{Transposition $(12)$
$e_1\leftrightarrow e_2$, $e_3$ fixed.}
\begin{align}
    (12)\cdot u_1
    &= \tfrac{1}{\sqrt{2}}(e_2-e_1)
    = -u_1,\\
    (12)\cdot u_2
    &= \tfrac{1}{\sqrt{6}}(e_2+e_1-2e_3)
    = u_2.
\end{align}
\begin{equation}
    \rho_3((12))
    = \begin{pmatrix}-1&0\\0&1\end{pmatrix}.
\end{equation}

\noindent\textit{Transposition $(13)$
$e_1\leftrightarrow e_3$, $e_2$ fixed.}
\begin{align}
    (13)\cdot u_1
    &= \tfrac{1}{\sqrt{2}}(e_3-e_2),\\
    (13)\cdot u_2
    &= \tfrac{1}{\sqrt{6}}(e_3+e_2-2e_1).
\end{align}
Computing matrix entries are
\begin{align}
    [\rho_3((13))]_{11}
    &= u_1^T\cdot\tfrac{1}{\sqrt{2}}(e_3-e_2)
    = \tfrac{1}{\sqrt{2}}\cdot\tfrac{1}{\sqrt{2}}
    (1,-1,0)\cdot(0,-1,1)
    = \tfrac{1}{2}(0+1+0)=\tfrac{1}{2},\\
    [\rho_3((13))]_{21}
    &= u_2^T\cdot\tfrac{1}{\sqrt{2}}(e_3-e_2)
    = \tfrac{1}{\sqrt{2}}\cdot\tfrac{1}{\sqrt{6}}
    (1,1,-2)\cdot(0,-1,1)
    = \tfrac{1}{\sqrt{12}}(0-1-2)
    = -\tfrac{\sqrt{3}}{2},\\
    [\rho_3((13))]_{12}
    &= u_1^T\cdot\tfrac{1}{\sqrt{6}}(e_3+e_2-2e_1)
    = \tfrac{1}{\sqrt{6}}\cdot\tfrac{1}{\sqrt{2}}
    (1,-1,0)\cdot(-2,1,1)
    = \tfrac{1}{\sqrt{12}}(-2-1+0)
    = -\tfrac{\sqrt{3}}{2},\\
    [\rho_3((13))]_{22}
    &= u_2^T\cdot\tfrac{1}{\sqrt{6}}(e_3+e_2-2e_1)
    = \tfrac{1}{\sqrt{6}}\cdot\tfrac{1}{\sqrt{6}}
    (1,1,-2)\cdot(-2,1,1)
    = \tfrac{1}{6}(-2+1-2)
    = -\tfrac{1}{2}.
\end{align}
\begin{equation}
    \rho_3((13))
    = \begin{pmatrix}
    \frac{1}{2}&-\frac{\sqrt{3}}{2}\\
    -\frac{\sqrt{3}}{2}&-\frac{1}{2}
    \end{pmatrix}.
\end{equation}

\noindent\textit{Transposition $(23)$
$e_2\leftrightarrow e_3$, $e_1$ fixed.}

By the same procedure
\begin{equation}
    \rho_3((23))
    = \begin{pmatrix}
    \frac{1}{2}&\frac{\sqrt{3}}{2}\\
    \frac{\sqrt{3}}{2}&-\frac{1}{2}
    \end{pmatrix}.
\end{equation}

\noindent\textit{3-cycle $(123)$:
$e_1\mapsto e_2\mapsto e_3\mapsto e_1$.}
\begin{align}
    (123)\cdot u_1
    &= \tfrac{1}{\sqrt{2}}(e_2-e_3),\\
    (123)\cdot u_2
    &= \tfrac{1}{\sqrt{6}}(e_2+e_3-2e_1).
\end{align}
Computing matrix entries
\begin{align}
    [\rho_3((123))]_{11}
    &= \tfrac{1}{\sqrt{2}}\cdot\tfrac{1}{\sqrt{2}}
    (1,-1,0)\cdot(0,1,-1)
    = \tfrac{1}{2}(0-1+0) = -\tfrac{1}{2},\\
    [\rho_3((123))]_{21}
    &= \tfrac{1}{\sqrt{2}}\cdot\tfrac{1}{\sqrt{6}}
    (1,1,-2)\cdot(0,1,-1)
    = \tfrac{1}{\sqrt{12}}(0+1+2) = \tfrac{\sqrt{3}}{2},\\
    [\rho_3((123))]_{12}
    &= \tfrac{1}{\sqrt{6}}\cdot\tfrac{1}{\sqrt{2}}
    (1,-1,0)\cdot(-2,1,1)
    = \tfrac{1}{\sqrt{12}}(-2-1+0) = -\tfrac{\sqrt{3}}{2},\\
    [\rho_3((123))]_{22}
    &= \tfrac{1}{\sqrt{6}}\cdot\tfrac{1}{\sqrt{6}}
    (1,1,-2)\cdot(-2,1,1)
    = \tfrac{1}{6}(-2+1-2) = -\tfrac{1}{2}.
\end{align}
\begin{equation}
    \rho_3((123))
    = \begin{pmatrix}
    -\frac{1}{2}&-\frac{\sqrt{3}}{2}\\
    \frac{\sqrt{3}}{2}&-\frac{1}{2}
    \end{pmatrix}.
\end{equation}

\noindent\textit{3-cycle $(132)$}
Since $(132)=(123)^{-1}$ and the representation
is unitary, $\rho_3((132))=\rho_3((123))^T$:
\begin{equation}
    \rho_3((132))
    = \begin{pmatrix}
    -\frac{1}{2}&\frac{\sqrt{3}}{2}\\
    -\frac{\sqrt{3}}{2}&-\frac{1}{2}
    \end{pmatrix}.
\end{equation}

\subsection*{A.6 Character table of $S_3$}

The character $\chi_\mu(g) = \mathrm{tr}(\rho_\mu(g))$
is the trace of the representation matrix. Since
conjugate elements always have the same trace
($\mathrm{tr}(ABA^{-1})=\mathrm{tr}(B)$), the
character is constant on each conjugacy class.

\noindent\textbf{Irrep 1:} $\chi_1(g) = \mathrm{tr}(1) = 1$
for all $g$.

\noindent\textbf{Irrep 2:} $\chi_2(g) = \mathrm{sgn}(g)$:
$+1$ for $e$ and 3-cycles, $-1$ for transpositions.

\noindent\textbf{Irrep 3:} $\chi_3(g) = \mathrm{tr}(\rho_3(g))$
\begin{align}
    \chi_3(e) &= 1+1 = 2,\\
    \chi_3((12)) &= -1+1 = 0,\quad
    \chi_3((13)) = \tfrac{1}{2}-\tfrac{1}{2}=0,\quad
    \chi_3((23))=\tfrac{1}{2}-\tfrac{1}{2}=0,\\
    \chi_3((123)) &= -\tfrac{1}{2}-\tfrac{1}{2}=-1,\quad
    \chi_3((132))=-\tfrac{1}{2}-\tfrac{1}{2}=-1.
\end{align}

The complete character table of $S_3$
\begin{equation}
\label{tab:S3_chars}
\begin{array}{c|ccc}
\text{Irrep} & C_1\ (1\text{ element})
& C_2\ (3\text{ elements})
& C_3\ (2\text{ elements})\\
\hline
\text{Irrep 1}\ (d=1) & 1 & 1 & 1\\
\text{Irrep 2}\ (d=1) & 1 & -1 & 1\\
\text{Irrep 3}\ (d=2) & 2 & 0 & -1
\end{array}
\end{equation}

\subsection*{A.7 Irreps of $\Gamma_A = S_3^A\times S_3^B$}

The irreps of a direct product group are all pairs
$(\mu_A,\mu_B)$ of irreps of the factors, with
\begin{equation}
    d_{(\mu_A,\mu_B)} = d_{\mu_A}\cdot d_{\mu_B},
    \qquad
    \chi_{(\mu_A,\mu_B)}(g_A,g_B)
    = \chi_{\mu_A}(g_A)\cdot\chi_{\mu_B}(g_B).
\end{equation}
There are $3\times3=9$ irreps of $\Gamma_A$, with
dimensions:
\begin{equation}
\begin{array}{c|ccc}
(\mu_A,\mu_B) & \mu_B=1 & \mu_B=2 & \mu_B=3\\
\hline
\mu_A=1 & 1 & 1 & 2\\
\mu_A=2 & 1 & 1 & 2\\
\mu_A=3 & 2 & 2 & 4
\end{array}
\end{equation}

\subsection*{A.8 Permutation characters on
$\mathcal{H}_A$ and $\mathcal{H}_B$}

$\mathcal{H}_A = (\mathbb{C}^2)^{\otimes 3}$
has dimension $2^3=8$. The permutation character
$\chi_{\mathrm{perm}}^A(g_A,g_B)$ counts
$A$-basis states $\alpha\in\{0,1\}^3$ fixed by
$(g_A,g_B)$. Since $g_B$ permutes $B$-sites which
are absent from $\mathcal{H}_A$, it acts as the
identity on $\mathcal{H}_A$:
\begin{equation}
    \chi_{\mathrm{perm}}^A(g_A,g_B)
    = \chi_{\mathrm{perm}}^A(g_A).
\end{equation}
A basis state $\alpha=(s_1,s_2,s_3)\in\{0,1\}^3$
is fixed by $g_A$ iff $s_{g_A(i)}=s_i$ for all $i$.

\noindent\textit{Class $C_1$ (identity):}
Every state is fixed.
$\chi_{\mathrm{perm}}^A(C_1)=8$.

\noindent\textit{Class $C_2$ (transpositions):}
Take $(12)$: state fixed iff $s_1=s_2$.
Fixed states: $(0,0,0),(0,0,1),(1,1,0),(1,1,1)$.
Count: 4.
$\chi_{\mathrm{perm}}^A(C_2)=4$.

\noindent\textit{Class $C_3$ (3-cycles)}
Take $(123)$: state fixed iff $s_1=s_2=s_3$.
Fixed states: $(0,0,0),(1,1,1)$.
Count: 2.
$\chi_{\mathrm{perm}}^A(C_3)=2$.

By symmetry of the balanced bipartition, the same
values hold for $\mathcal{H}_B$:
$\chi_{\mathrm{perm}}^B(C_1)=8$,
$\chi_{\mathrm{perm}}^B(C_2)=4$,
$\chi_{\mathrm{perm}}^B(C_3)=2$.

\subsection*{A.9 Multiplicities $m_{(\mu_A,\mu_B)}^A$}

The multiplicity of irrep $(\mu_A,\mu_B)$ of
$\Gamma_A$ in $\mathcal{H}_A$ is
\begin{equation}
    m_{(\mu_A,\mu_B)}^A
    = \frac{1}{|\Gamma_A|}
    \sum_{(g_A,g_B)\in\Gamma_A}
    \chi_{(\mu_A,\mu_B)}(g_A,g_B)^*\,
    \chi_{\mathrm{perm}}^A(g_A,g_B).
\end{equation}
Substituting the product form of the character
and using conjugacy class sizes and the character table, we obtain for $F_{\mu_A}\left(
    \sum_{g_A\in S_3^A}
    \chi_{\mu_A}(g_A)^*\,
    \chi_{\mathrm{perm}}^A(g_A)
    \right)$
\begin{align}
    F_1 &= 1\cdot\chi_1(C_1)\cdot8
    + 3\cdot\chi_1(C_2)\cdot4
    + 2\cdot\chi_1(C_3)\cdot2\nonumber\\
    &= 1\cdot1\cdot8+3\cdot1\cdot4+2\cdot1\cdot2
    = 8+12+4 = 24,\\
    F_2 &= 1\cdot1\cdot8+3\cdot(-1)\cdot4+2\cdot1\cdot2
    = 8-12+4 = 0,\\
    F_3 &= 1\cdot2\cdot8+3\cdot0\cdot4+2\cdot(-1)\cdot2
    = 16+0-4 = 12.
\end{align}

 Similarly for
$G_{\mu_B} = \sum_{g_B\in S_3^B}\chi_{\mu_B}(g_B)^*$
is the sum of the character over all group elements
\begin{align}
    G_1 &= 1\cdot1+3\cdot1+2\cdot1 = 6,\\
    G_2 &= 1\cdot1+3\cdot(-1)+2\cdot1 = 0,\\
    G_3 &= 1\cdot2+3\cdot0+2\cdot(-1) = 0.
\end{align}

\subsection*{A.11 The automorphism bound}

We now have all the ingredients for
Theorem~\ref{thm:main}. The complete table of
irrep dimensions and multiplicities is

\begin{table}[h]
\centering
\begin{tabular}{ccccc}
\hline
$(\mu_A,\mu_B)$ & $d_\mu$ &
$m_\mu^A$ & $m_\mu^B$ &
$d_\mu\min(m_\mu^A,m_\mu^B)$ \\
\hline
$(1,1)$ & 1 & 4 & 4 & 4 \\
$(1,2)$ & 1 & 0 & 0 & 0 \\
$(1,3)$ & 2 & 0 & 2 & 0 \\
$(2,1)$ & 1 & 0 & 0 & 0 \\
$(2,2)$ & 1 & 0 & 0 & 0 \\
$(2,3)$ & 2 & 0 & 0 & 0 \\
$(3,1)$ & 2 & 2 & 0 & 0 \\
$(3,2)$ & 2 & 0 & 0 & 0 \\
$(3,3)$ & 4 & 0 & 0 & 0 \\
\hline
Total & & & & 4 \\
\hline
\end{tabular}
\caption{Complete irrep data for $K_{3,3}$.
}
\label{tab:K33_full}
\end{table}

Therefore
\begin{equation}
    \sum_\mu d_\mu\min(m_\mu^A,m_\mu^B)
    = 1\times\min(4,4) = 4,
\end{equation}
and Theorem~\ref{thm:main} yields
\begin{equation}
    S_{N/2}(K_{3,3}^{XY})
    \leq \log 4 = 2\log 2.
\end{equation}

\bigskip
\end{document}